\documentclass[preprint,
UKenglish]{eptcs}

\usepackage{amsmath,amsthm,amssymb}
\usepackage{xspace,enumerate,color,epsfig} 
\usepackage{float}
\usepackage{graphicx}

\usepackage{stmaryrd}
\usepackage{mathtools}
\usepackage{microtype}
\usepackage{multirow}
\usepackage{url}

\usepackage[utf8]{inputenc}
\usepackage{scrextend}
\usepackage[english]{babel}

\usepackage{xcolor}
\definecolor{zx_green}{rgb}{216,248,216}
\definecolor{zx_red}{rgb}{232,165,165}

\usepackage{tikzit}

\tikzstyle{box}=[shape=rectangle, text height=1.5ex, text depth=0.25ex, yshift=0.5mm, fill=white, draw=black, minimum height=5mm, yshift=-0.5mm, minimum width=5mm, font={\small}]
\tikzstyle{Z dot}=[inner sep=0mm, minimum size=2mm, shape=circle, draw=black, fill={rgb,255: red,216; green,248; blue,216}, tikzit fill={rgb,255: red,216; green,248; blue,216}]
\tikzstyle{Z phase dot}=[minimum size=4.75mm, font={\footnotesize}, shape=rectangle, rounded corners=1.9mm, inner sep=0.1mm, outer sep=-2mm, scale=0.8, tikzit shape=circle, draw=black, fill={rgb,255: red,216; green,248; blue,216}, tikzit draw=blue, tikzit fill={rgb,255: red,216; green,248; blue,216}]
\tikzstyle{X dot}=[Z dot, shape=circle, draw=black, fill={rgb,255: red,232; green,165; blue,165}, tikzit fill={rgb,255: red,232; green,165; blue,165}]
\tikzstyle{X phase dot}=[Z phase dot, tikzit shape=circle, tikzit fill={rgb,255: red,232; green,165; blue,165}, fill={rgb,255: red,232; green,165; blue,165}, font={\footnotesize}, tikzit draw=blue]
\tikzstyle{XD dot}=[Z dot, fill={rgb,255: red,255; green,200; blue,240}, tikzit fill={rgb,255: red,255; green,200; blue,240}]
\tikzstyle{XD phase dot}=[Z phase dot, fill={rgb,255: red,255; green,200; blue,240}, tikzit fill={rgb,255: red,255; green,200; blue,240}, tikzit draw=blue]
\tikzstyle{hadamard}=[fill=yellow, draw=black, shape=rectangle, inner sep=0.6mm, minimum height=1.5mm, minimum width=1.5mm, xslant=0.5]
\tikzstyle{hz}=[hadamard, fill={rgb,255: red,216; green,248; blue,216}, shape=rectangle, tikzit fill={rgb,255: red,216; green,248; blue,216}, minimum height=2 mm, minimum width=1.25 mm, tikzit draw=black]
\tikzstyle{hx}=[hadamard, fill={rgb,255: red,232; green,165; blue,165}, shape=rectangle, tikzit fill={rgb,255: red,232; green,165; blue,165}, minimum height=2 mm, minimum width=1.25 mm, tikzit draw=black]
\tikzstyle{vertex}=[inner sep=0mm, minimum size=1mm, shape=circle, draw=black, fill=black]
\tikzstyle{vertex set}=[inner sep=0mm, minimum size=1mm, shape=circle, draw=black, fill=white, font={\footnotesize\boldmath}]
\tikzstyle{meter}=[draw, fill=white, minimum width=2em, minimum height=1.5em, rectangle, path picture={\draw ([shift={(.1,.24)}]path picture bounding box.south west) to[bend left=50] ([shift={(-.1,.24)}]path picture bounding box.south east);\draw[-{Latex[scale=0.6]}] ([shift={(0,.1)}]path picture bounding box.south) -- ([shift={(.3,-.1)}]path picture bounding box.north);}, tikzit shape=rectangle]
\tikzstyle{white dot}=[Z dot]
\tikzstyle{gray dot}=[X dot]
\tikzstyle{white phase dot}=[Z phase dot]
\tikzstyle{gray phase dot}=[X phase dot]
\tikzstyle{red ket}=[fill={rgb,255: red,232; green,165; blue,165}, draw=black, shape=isosceles triangle, tikzit fill={rgb,255: red,232; green,165; blue,165}, tikzit draw=black, inner sep=0 mm, outer sep=2 mm]
\tikzstyle{green ket}=[fill={rgb,255: red,216; green,248; blue,216}, draw=black, shape=isosceles triangle, tikzit fill={rgb,255: red,216; green,248; blue,216}, tikzit draw=black, inner sep=0 mm, outer sep=2 mm]
\tikzstyle{filament}=[hadamard, fill=yellow, draw=none, minimum height=0.01mm]
\tikzstyle{gate}=[box, minimum height=10mm, minimum width=10mm]
\tikzstyle{2 control}=[vertex set, draw=blue, inner sep=0.5pt]
\tikzstyle{1 control}=[2 control, draw=red]
\tikzstyle{0 control}=[2 control, draw=black]
\tikzstyle{small dot}=[vertex, minimum size=1 mm, draw=black, tikzit draw=black, tikzit fill=black, tikzit shape=circle]
\tikzstyle{tallbox}=[box, minimum height=12mm]
\tikzstyle{targ}=[vertex set, minimum size=0.5mm, inner sep=-0.5mm, tikzit shape=circle, shape=circle, tikzit draw=black]
\tikzstyle{hadamardbox}=[hadamard, xslant=0]
\tikzstyle{zn}=[shape=zn, tikzit draw=black, draw=black, inner sep=2pt]

\tikzstyle{directedarrow}=[draw={rgb,255: red,223; green,223; blue,223}, ->, tikzit draw={rgb,255: red,223; green,223; blue,223}, line width=1 pt]
\tikzstyle{simple}=[-]
\tikzstyle{hadamard edge}=[-, color={rgb,255: red,0; green,100; blue,248}, dashed, dash pattern=on 2pt off 0.7pt, tikzit draw={rgb,255: red,0; green,100; blue,248}]
\tikzstyle{brace edge}=[-, tikzit draw=blue, decorate, decoration={brace,amplitude=1mm,raise=-1mm}]
\tikzstyle{gray}=[-, draw={rgb,255: red,223; green,223; blue,223}, line width=1 pt]
\tikzstyle{arrow}=[<-, draw={rgb,255: red,128; green,128; blue,128}]
\tikzstyle{double-arrow}=[draw={rgb,255: red,128; green,128; blue,128}, <->]
\tikzstyle{dashed edge}=[-, dashed, dash pattern=on 2pt off 0.5pt, draw=black]
\tikzstyle{diredge}=[->]
\tikzstyle{double edge}=[-, double, shorten <=-1mm, shorten >=-1mm, double distance=2pt]
\tikzstyle{thin}=[-, line width=0.05mm]
\tikzstyle{thin gray}=[-, draw={rgb,255: red,223; green,223; blue,223}, line width=0.05mm]
\tikzstyle{less thin}=[-, line width=0.1mm]
\tikzstyle{dashed gray edge}=[-, dashed edge, draw={rgb,255: red,128; green,128; blue,128}]
\tikzstyle{light right directed arrow}=[->, directedarrow, draw={rgb,255: red,223; green,223; blue,223}, line width=0.2mm]
\tikzstyle{diredge0.3}=[->, line width=0.3 mm]
\tikzstyle{less thin gray}=[-, draw={rgb,255: red,223; green,223; blue,223}]
\tikzstyle{dashed thin purple}=[-, dashed, line width=0.1mm, draw={rgb,255: red,128; green,106; blue,219}]
\tikzstyle{hadamardedge}=[-, color={rgb,255: red,100; green,200; blue,248}, dashed, dash pattern=on 2pt off 0.7pt, tikzit draw={rgb,255: red,120; green,220; blue,248}]
\tikzstyle{line0.3}=[-, line width=0.3mm]
\tikzstyle{light blue line}=[-, color={rgb,255: red,100; green,200; blue,248}, tikzit draw={rgb,255: red,100; green,200; blue,248}]
\tikzstyle{pink0.2}=[-, color={rgb,255: red,248; green,100; blue,200}, line width=0.2mm, tikzit draw={rgb,255: red,248; green,100; blue,200}]

\input{zx.tikzdefs}

\newcommand{\R}{\mathbb{R}}

\newcommand{\coloneqq}{:=}

\theoremstyle{definition}
\newtheorem{theorem}{Theorem}[section]
\newtheorem{corollary}[theorem]{Corollary}
\newtheorem{lemma}[theorem]{Lemma}
\newtheorem{proposition}[theorem]{Proposition}

\newtheorem{definition}[theorem]{Definition}

\newtheorem{example*}[theorem]{Example*}
\newtheorem{examples*}[theorem]{Examples*}
\newtheorem{remark}[theorem]{Remark}
\newtheorem{remark*}[theorem]{Remark*}

\title{Building Qutrit Diagonal Gates from Phase Gadgets}
\def\titlerunning{Building Qutrit Diagonal Gates from Phase Gadgets}

\author{John van de Wetering
\institute{University of Oxford}
\email{john@vdwetering.name} \and
Lia Yeh
\institute{University of Oxford}
\email{lia.yeh@cs.ox.ac.uk} 
}
\def\authorrunning{J.~van de Wetering \& L.~Yeh}

\begin{document}

\maketitle

\begin{abstract}
    Phase gadgets have proved to be an indispensable tool for reasoning about ZX-diagrams, being used in optimisation and simulation of quantum circuits and the theory of measurement-based quantum computation. In this paper we study phase gadgets for qutrits. We present the \emph{flexsymmetric} variant of the original qutrit ZX-calculus, which allows for rewriting that is closer in spirit to the original (qubit) ZX-calculus. In this calculus phase gadgets look as you would expect, but there are non-trivial differences in their properties. We devise new qutrit-specific tricks to extend the graphical Fourier theory of qubits, resulting in a translation between the `additive' phase gadgets and a `multiplicative' counterpart we dub \emph{phase multipliers}.

    This enables us to generalise the qubit notion of multiple-control to qutrits in two ways.  The first type is controlling on a single tritstring, while the second type applies the gate a number of times equal to the tritwise multiplication modulo 3 of the control qutrits.
    We show how both types of control can be implemented for any qutrit $Z$ or $X$ phase gate, ancilla-free, and using only Clifford and phase gates. The first requires a polynomial number of gates and exponentially small phases, while the second requires an exponential number of gates, but constant sized phases. This is interesting, because such a construction is not possible in the qubit setting. 

    As an application of these results we find a construction for \emph{emulating} arbitrary qubit diagonal unitaries, and specifically find an ancilla-free emulation for the qubit CCZ gate that only requires three single-qutrit non-Clifford gates --- provably lower than the four $T$ gates needed for qubits with ancilla.
\end{abstract}

\section{Introduction}

Most quantum computing theory developed thus far has focussed on qubits --- two-level quantum systems.
However, there has been a recent surge of interest in studying the more general case of $d$-level quantum systems, called \emph{qudits}. This has led to applications of qudits for quantum algorithms~\cite{WangY2020quditsreview}, improving magic state distillation noise thresholds~\cite{CampbellE2014quditmsdthresholds}, and communication noise resilience~\cite{CozzolinoD2019quditcommunication}.
Qudits have been experimentally demonstrated on quantum processors based on ion traps~\cite{RingbauerM2021quditions} and superconducting devices~\cite{BlokM2021scrambling,YeB2018cphasephoton,YurtalanM2020Walsh-Hadamard,HillA2021doublycontrolled}.

The specific case of qu\emph{trits}, where $d=3$, has been used to improve qubit readout~\cite{MalletF2009qubitreadout2state}, but most notably, qutrits have been used to study \emph{emulation}: where qubit computation is emulated inside a subspace of the qudits to enable more resource-efficient gate implementations.  In contrast, it has been argued that qubits cannot simulate qudit (where $d > 2$ and $d$ is not a power of $2$) computation efficiently~\cite{BullockS2005qubitemuqudit}.

Much work on qutrits and emulation has focussed on \emph{classical} functions: those that come from a map of classical trits.
For instance, using qutrits we can build logarithmic-depth Toffolis~\cite{GokhaleP2019asymptotic,NikolaevaAS2022mctqutrit} and binary AND gates on superconducting qutrits~\cite{ChuJ2021superconductingand}.
This leaves open the question of whether there are any advantages to emulation by studying `truly' quantum gates such as diagonal unitaries.
For qubits a useful tool for understanding diagonal unitaries has been the concept of a \emph{phase gadget}~\cite{KissingerA2020reduc}. This is a type of symmetric multi-qubit interaction that occurs naturally in many hardware architectures~\cite{PinoJM2021honeywellarchitecture,PetitL2020silicontwoqubitgates,SheldonS2016IBMCRgate}, and serves as a good basis for optimising quantum circuits~\cite{phaseGadgetSynth,cowtan2020generic,deBeaudrapN2020reducepifourphase,deBeaudrapN2020treducspidernest,vandeWeteringJ2021globalgates,Backens2020extraction}. Any diagonal qubit unitary can be expressed as a product of phase gadgets by writing the unitary as a \emph{phase polynomial}~\cite{AmyVerification,deGriendA2020architecturephasepoly}.

In this paper we study the generalisation of phase gadgets to the qutrit setting. We do this by adapting the qutrit ZX-calculus of Refs.~\cite{GongX2017equivalence,WangQ2018qutrit} and transforming it into a \emph{flexsymmetric} calculus~\cite{Carette2021OTM} where the spiders have more desirable symmetry properties. We find this calculus has a simple set of rules for the Clifford fragment. We define phase gadgets analogously to the qubit case, meaning that as diagrams they look nearly identical. There are however significant differences between the qubit and qutrit gadgets. 
We will show that we can nevertheless use qutrit phase gadgets to construct some useful qutrit diagonal unitaries, such as controlled phase gates, and a type of gate we dub a \emph{phase multiplier}. This last one is possible by generalising the formula that leads to the \emph{graphical Fourier theory} for qubit diagonal unitaries~\cite{GraphicalFourier2019}.

As an application of our results we show how we can emulate an arbitrary qubit diagonal unitary using qutrit phase gadgets. This leads us to a construction of the emulated qubit CCZ gate that requires only three non-Clifford qutrit \emph{$R$ gates}~\cite{GlaudellA2022qutritmetaplecticsubset}. This is surprising because using just qubits, we would require at least four $T$ gates to implement the CCZ~\cite{HowardM2017resourcetheorymagic}.

We start the paper by reviewing the basics of qutrit quantum computation in Section~\ref{sec:qutrit-computing}. Then we introduce the flexsymmetric qutrit ZX-calculus in Section~\ref{sec:qutrit-ZX}. Diagonal qutrit unitaries, phase gadgets, controlled phase gates, and phase multipliers are studied in Section~\ref{sec:diagonal-gates}. We show how to use these to emulate diagonal qubit unitaries in Section~\ref{sec:emulation-application} and end with some discussion on future work in Section~\ref{sec:conclusion}.

\section{Qutrit quantum computation}\label{sec:qutrit-computing}

A qubit is a two-dimensional Hilbert space. Similarly, a qutrit is a three-dimensional Hilbert space. We will write $\ket{0}$, $\ket{1}$, and $\ket{2}$ for the standard computational basis states of a qutrit.
Any normalised qutrit state can then be written as $\ket{\psi} = \alpha \ket{0} +  \beta \ket{1} + \gamma \ket{2}$ where $\alpha,\beta,\gamma\in \mathbb{C}$ and $|\alpha|^2 + |\beta|^2 + |\gamma|^2 = 1$.

Several concepts for qubits extend to qutrits, or more generally to qu\emph{dits}, which are $d$-dimensional quantum systems. In particular, the concept of Pauli's and Cliffords.
For a $d$-dimensional qudit, we define the respective Pauli $X$ and $Z$ gates as
\begin{equation}
    X\ket{k} = \ket{k+1} \qquad\qquad Z\ket{k} = \omega^k \ket{k}
\end{equation}
where $\omega:= e^{2\pi i/d}$ is such that $\omega^d = 1$, and the addition $\ket{k+1}$ is taken modulo $d$~\cite{GottesmanD1999ftqudit,HowardM2012quditTgate}. Note that for qubits this $X$ gate is just the NOT gate, while $Z=\text{diag}(1,-1)$.
We call unitaries generated by products and tensor products of the $X$ and $Z$ gate \emph{Pauli} gates. 
 In this paper we will work solely with qutrits, so we take $\omega$ to always be equal to $e^{2\pi i/3}$. Note that $\omega^{-1} = \omega^2 = \bar{\omega}$ where $\bar{z}$ denotes the complex conjugate of $z$.

For a qubit there is only one non-trivial permutation of the standard basis states, implemented by the $X$ gate.
For qutrits there are five non-trivial permutations of the basis states. By analogy we will call them all ternary $X$ gates. These gates are $X_{+1}$, $X_{-1}$, $X_{01}$, $X_{12}$, and $X_{02}$. 
The gate $X_{\pm 1}$ sends $\ket{t}$ to $\ket{(t \pm 1) \text{ mod } 3}$ for $t \in \{0, 1, 2\}$; $X_{01}$ is just the qubit $X$ gate which is the identity when the input is $\ket{2}$; $X_{12}$ sends $\ket{1}$ to $\ket{2}$ and $\ket{2}$ to $\ket{1}$, and likewise for $X_{02}$.
Note that the qutrit Pauli $X$ gate is the $X_{+1}$ gate, while $X^\dagger = X_{-1} = X^2$.

Another concept that translates to qutrits (or more generally qudits) is that of Clifford unitaries.

\begin{definition}\label{def:Clifford}
    Let $U$ be a qudit unitary acting on $n$ qudits. We say it is \emph{Clifford} when every Pauli is mapped to another Pauli under conjugation by $U$. I.e.~if $UPU^\dagger$ is a Pauli for any Pauli $P$.
\end{definition}

The set of $n$-qudit Cliffords forms a group under composition. For qubits, this group is generated by the $S$, Hadamard and CX gates. The same is true for qutrits, for the right generalisation of these gates\footnote{The gate definitions for various qudit Cliffords may vary across the literature up to a global phase.  Indeed, by Definition~\ref{def:Clifford}, whether a gate is Clifford is invariant under changes in global phase.}~\cite{GottesmanD1999ftqudit}.

\begin{definition}
    The qutrit $S$ gate is $S\coloneqq \text{diag}(1,1,\omega)$. I.e.~it multiplies the $\ket{2}$ state by the phase $\omega$.
\end{definition}

For qubits, the Hadamard gate interchanges the $Z$ eigenbasis $\{\ket{0},\ket{1}$ and the $X$ eigenbasis consisting of the states $\ket{\pm} \coloneqq \frac{1}{\sqrt{2}}(\ket{0}\pm \ket{1})$. The same holds for the qutrit Hadamard.
In this case the $X$ basis consists of the following states:
\begin{equation*}
    \ket{+} \ \coloneqq \ \frac{1}{\sqrt{3}} (\ket{0}+\ket{1}+\ket{2}) \qquad
    \ket{\omega} \coloneqq\  \frac{1}{\sqrt{3}} (\ket{0}+\omega\ket{1}+\bar{\omega}\ket{2}) \qquad
    \ket{\bar{\omega}} \ \coloneqq\  \frac{1}{\sqrt{3}} (\ket{0}+\bar{\omega}\ket{1}+\omega\ket{2})
\end{equation*}

\begin{definition}
    The \emph{qutrit Hadamard gate} $H$ is the unitary mapping $\ket{0} \mapsto \ket{+}$, $\ket{1}\mapsto \ket{\omega}$ and $\ket{2} \mapsto \ket{\bar{\omega}}$. 
    \begin{equation}\label{eq:hgatedef}
        H \ \coloneqq \ \frac{1}{\sqrt{3}}\begin{pmatrix}
            1 & 1 & 1 \\
            1 & \omega & \bar{\omega} \\
            1 & \bar{\omega} & \omega
        \end{pmatrix}
    \end{equation}
\end{definition}
Note that, unlike the qubit Hadamard, the qutrit Hadamard is \emph{not} self-inverse. 
Instead we have $H^2 = X_{12}$ so that $H^4 = \mathbb{I}$. This means that $H^\dagger = H^3$.

The final Clifford gate we need is the qutrit CX.
\begin{definition}
    The qutrit CX is defined such that $\text{CX}\ket{i,j} = \ket{i,(i+j)~\text{mod }3}$, where $i,j\in \{0,1,2\}$.
\end{definition}
Any qutrit Clifford unitary can be written as a composition of $S$, $H$ and CX gates (up to global phase).
Clifford gates are efficiently classically simulable, so we need to add a non-Clifford gate to get an (approximately) universal gate set for quantum computing~\cite{GottesmanD1999ftqudit}. Here we consider when this is a phase gate.
\begin{definition}\label{def:Z-phase-gate}
    We write $Z(a,b)$ for the \emph{phase gate} that acts as $Z(a,b)\ket{0} = \ket{0}$, $Z(a,b)\ket{1} = \omega^a\ket{1}$ and $Z(a,b)\ket{2} = \omega^b\ket{2}$ where we take $a,b\in \mathbb{R}$.
\end{definition}
We define $Z(a,b)$ in this way, taking $a$ and $b$ to correspond to powers of $\omega$, because then $Z(a,b)$ is Clifford iff $a$ and $b$ are both integers, so that we can easily see from the parameters whether the gate is Clifford or not. The group of $Z(a,b)$ phase gates constitutes the group of diagonal single-qutrit unitaries modded out by a global phase. Composition of these gates is given by $Z(a,b)\cdot Z(c,d)=Z(a+c,b+d)$. Note that $S = Z(0,1)$.
This brings us to the definition of the qutrit $T$ gate.
\begin{definition}
    The qutrit $T$ gate is defined as $T \coloneqq Z(\frac{1}{3}, -\frac{1}{3}) = \text{diag}(1,\omega^{\frac{1}{3}},\omega^{-\frac{1}{3}})$~\cite{PrakashS2018normalform,CampbellE2012tgatedistillation,HowardM2012quditTgate}.
\end{definition}
Like the qubit $T$ gate, the qutrit $T$ gate belongs to the third level of the Clifford hierarchy, can be injected into a circuit using magic states, and its magic states can be distilled by magic state distillation. This means that we can fault-tolerantly implement this qutrit $T$ gate on many types of quantum error correcting codes.
Also as for qubits, the qutrit Clifford+$T$ gate set is approximately universal, meaning that we can approximate any qutrit unitary using just Clifford gates and the $T$ gate~\cite[Theorem 1]{CuiS2015universalmetaplectic}.

There is another useful single-qutrit non-Clifford gate.
\begin{definition}\label{def:Rgate}
The qutrit \emph{reflection} gate is defined as $R\coloneqq Z(0,3/2) = \text{diag}(1,1,-1)$.
\end{definition}
Like the $T$ gate, the $R$ gate can be added to the Clifford gate set to attain universality~\cite{GottesmanD1999ftqudit}, as explicitly proved in Ref.~\cite[Theorem~2]{CuiS2015universalmetaplectic}.
It can be exactly synthesized fault-tolerantly in three known ways: magic state distillation followed by repeat-until-success injection~\cite{AnwarH2012r2distillation}, braiding and topological measurement of weakly-integral non-Abelian anyons~\cite{CuiS2015universalweakly,CuiS2015universalmetaplectic} followed by repeat-until-success injection~\cite{AnwarH2012r2distillation}, or unitarily in qutrit Clifford+$T$~\cite{GlaudellA2022qutritmetaplecticsubset}.

\subsection{Controlled unitaries}

When we have an $n$-qubit unitary $U$, we can speak of the controlled gate that implements $U$. This is the $(n+1)$-qubit gate that acts as the identity when the first qubit is in the $\ket{0}$ state, and implements $U$ on the last $n$ qubits if the first qubit is in the $\ket{1}$ state.
For qutrits there are multiple notions of control.

\begin{definition}\label{def:ket2-controlled}
Let $U$ be a qutrit unitary. Then the \emph{$\ket{2}$-controlled $U$} is the unitary that acts as
\begin{equation*}
    \ket{0}\otimes \ket{\psi} \mapsto \ket{0}\otimes \ket{\psi} \qquad
    \ket{1}\otimes \ket{\psi} \mapsto \ket{1}\otimes \ket{\psi} \qquad
    \ket{2}\otimes \ket{\psi} \mapsto \ket{2}\otimes U\ket{\psi}
\end{equation*}
I.e.~it implements $U$ on the last qutrits if and only if the first qutrit is in the $\ket{2}$ state. 
\end{definition}
Note that by conjugating the first qutrit with $X_{+1}$ or $X_{-1}$ gates we can make the gate also be controlled on the $\ket{1}$ or $\ket{0}$ state.
A different notion of qutrit control was introduced in Ref.~\cite{BocharovA2017ternaryshor} where if the control is in the $\ket{x}$ state, then it should apply $U^x$ on the target, i.e.~apply $U$ once iff $x=1$ and $U^2$ iff $x=2$.
An example of this is the Clifford CX gate defined earlier, which applies $X_{+1}$ when the control is $\ket{1}$ and $X_{+2}$ when it is $\ket{2}$.
Note that we can get this latter notion of control from the former: just apply a $\ket{1}$-controlled $U$, followed by a $\ket{2}$-controlled $U^2$. 

A number of Clifford+$T$ constructions for controlled qutrit unitaries are already known. For instance, all the $\ket{2}$-controlled permutation $X$ gates can be built from the constructions given in Ref.~\cite{BocharovA2016ternaryarithmetics}.
In our previous work, we provided ancilla-free explicit constructions for any multiple-controlled Clifford+$T$ unitary in the Clifford+$T$ gate set, with gate count polynomial in the number of controls~\cite{YehL2022qutritcontrolledcliffordplust}.
In this work, by using the qutrit ZX-calculus, we will build upon our previous results and show how to construct multiple-controlled phase gates for an arbitrary phase. 

\section{The qutrit ZX-calculus}\label{sec:qutrit-ZX}
We will assume the reader has some familiarity with the original qubit ZX-calculus~\cite{CoeckeB2011interacting}. For a review see Ref.~\cite{vandewetering2020zxcalculus}. 

A qutrit ZX-calculus was presented and used in Refs.~\cite{RanchinA2014quditzx,WangQ2014qutritcalculus,GongX2017equivalence,WangQ2018qutrit,townsend-teague2021classifying}. While quite similar to the qubit one, it loses some of the properties that make the original easy to work with. In particular, for each X-spider, the distinction between its input wires and output wires becomes important. This means we can no longer treat qutrit ZX-diagrams as undirected graphs with the spiders as vertices. This makes intuitive reasoning about these diagrams harder, and also complicates the implementation of software for dealing with these diagrams.

Here we will present a variation on the qutrit ZX-calculus of Refs.~\cite{GongX2017equivalence,WangQ2018qutrit} where the spiders do enjoy this additional symmetry between inputs and outputs. The way we do this is by redefining the X-spider. In the original qutrit ZX-calculus we have
\begin{equation}
\input{Xsp.tikz} \ \ \propto \sum_{\substack{\vec x, \vec y\\x_1+\cdots+x_n = y_1+\cdots + y_m}} \!\!\!\!\!\!\!\!\ketbra{\vec y}{\vec x}.
\end{equation}
Here the sum $x_1+\cdots+x_n = y_1+\cdots + y_m$ is taken modulo 3. If we put a cup on one of the wires to turn an output into an input, then this has the effect of introducing a minus sign on that variable, changing for instance $x_1+x_2 = y_1+y_2$ into $x_1+x_2-y_2 = y_1$. For qubits this is not a problem since $-x = x$ modulo $2$, but for qutrits this changes the map. We fix this by defining a new X-spider as
\begin{equation}\label{eq:XDsp}
\input{XDsp.tikz} \ \ \propto \sum_{\substack{\vec x, \vec y\\x_1+\cdots+x_n+y_1+\cdots + y_m = 0}} \!\!\!\!\!\!\!\!\ketbra{\vec y}{\vec x}.
\end{equation}
We see that in this definition the inputs and outputs are treated on equal footing. In order to prevent confusion with earlier work, we will denote this new X-spider in pink, instead of in red\footnote{We have checked the accessibility of this color scheme; in fact, a red-green colorblind person greatly preferred this pink to the default ZX red.}.

Let's now give the full definition of the spiders. We define the Z-spider as
\[\input{Zsp-phase.tikz} \ = \ \ketbra{0\cdots 0}{0\cdots 0} \ +\  \omega^\alpha \ketbra{1\cdots 1}{1\cdots 1} \ +\  \omega^\beta \ketbra{2\cdots 2}{2\cdots 2}.\]
Here we have two phase angles $\alpha$ and $\beta$, as opposed to just the one angle in qubit ZX. In general, for a $d$-dimensional spider, you will need to specify $d-1$ phases. In particular, when written in a spider $\frac{\alpha}{\beta}$ should be interpreted as two different phases and not as the fraction $\alpha/\beta$.
Note that we define the phase angles as $\omega^\alpha$ and $\omega^\beta$ so that these correspond to the complex phases $e^{\frac{2\pi i}{3} \alpha}$ and $e^{\frac{2\pi i}{3} \beta}$. This means that when $\alpha$ and $\beta$ are integers, that the spiders correspond to the Clifford fragment of the calculus.
We define the X-spider similarly, but with respect to the X-basis:
\[\input{XDsp-phase.tikz} \ = \ \ketbra{+\cdots +}{+\cdots +} \ +\  \omega^\alpha \ketbra{\bar{\omega}\cdots \bar{\omega}}{\omega\cdots \omega} \ +\  \omega^\beta \ketbra{\omega\cdots \omega}{\bar{\omega}\cdots \bar{\omega}}\]
This requires some explanation, because this does not look symmetric in the inputs and outputs. However, note that 
$\bra{\omega} = (\ket{\omega})^\dagger = (\ket{0}+\omega \ket{1}+\omega \ket{2})^\dagger = \bra{0} + \bar{\omega} \bra{1} + \omega \bra{2}$.
Hence, if we take the transpose of $\ket{\omega}$ we actually get $\bra{\bar{\omega}}$.
It is straightforward to check that with $\alpha=\beta=0$ we get back Eq.~\eqref{eq:XDsp}. These definitions of the Z-spider and X-spider satisfy the symmetry properties we want, namely:
\begin{equation*}
    \scalebox{0.9}{\input{spider-symmetries.tikz}}
\end{equation*}
These symmetries mean our spiders are \emph{flexsymmetric}, as defined by Carette~\cite{Carette2021OTM}, and as a result we may treat our ZX-diagrams as undirected graphs with the spiders as vertices. Note that here the cups and caps are defined with respect to the $Z$ basis: $\subset \ =\  \ket{00}+\ket{11}+\ket{22}$. As usual, our calculus also formally has generators for the identity wire and the swap.

It will be useful to introduce an additional graphical generator for the Hadamard:
\begin{equation}
\left\llbracket\begin{tikzpicture}
	\begin{pgfonlayer}{nodelayer}
		\node [style=hadamard] (0) at (0, 0) {};
		\node [style=none] (1) at (-0.75, 0) {};
		\node [style=none] (2) at (0.75, 0) {};
	\end{pgfonlayer}
	\begin{pgfonlayer}{edgelayer}
		\draw (1.center) to (0);
		\draw (2.center) to (0);
	\end{pgfonlayer}
\end{tikzpicture}
\right\rrbracket=\ket{+}\bra{0}+ \ket{\omega}\bra{1}+\ket{\bar{\omega}}\bra{2}=\ket{0}\bra{+}+ \ket{1}\bra{\bar{\omega}}+\ket{2}\bra{\omega}.
\end{equation}
We write the Hadamard as a slanted box, because it is self-transpose, but not self-adjoint, and so should be denoted in a way that is symmetric under a rotation, but not a reflection.

Our redefinition of the X-spider comes at a `cost'. Namely, the 1-input, 1-output X-spider is no longer the identity:
$
\begin{tikzpicture}
	\begin{pgfonlayer}{nodelayer}
		\node [style=none] (0) at (-0.75, 0) {};
		\node [style=XD dot] (1) at (0, 0) {};
		\node [style=none] (2) at (0.75, 0) {};
	\end{pgfonlayer}
	\begin{pgfonlayer}{edgelayer}
		\draw (0.center) to (1);
		\draw (1) to (2.center);
	\end{pgfonlayer}
\end{tikzpicture}
 \ = \ \ket{0}\bra{0}+ \ket{2}\bra{1}+\ket{1}\bra{2} \ = \ \ket{+}\bra{+}+\ket{\bar{\omega}}\bra{\omega}+ \ket{\omega}\bra{\bar{\omega}}
$.
This map is implementing $\ket{x} \mapsto \ket{-x}$ where $-x$ is taken modulo 3, and is equal to $X_{12}$. Additionally, the X-spider is not really a spider any more in the sense that it doesn't satisfy the standard spider-fusion equation. Instead it satisfies the `harvestman equation'~\cite{Carette2021OTM} that also holds for for instance the W-spider~\cite{hadzihasanovic2015diagrammatic} and H-box~\cite{EPTCS287.2}:
\[\input{XD_rules/spidernewprime.tikz}\]

In Figure~\ref{fig:qutritphaseexactflexsymmetricrules}, we present a full set of rewrite rules for this qutrit ZX-calculus.
We have accounted for the global phase for each rule here as a complex number, as those will be relevant to us.  Note however that the rewrite rules are not scalar-accurate as we are ignoring factors of $\sqrt{3}$.

\begin{figure}[!tb]
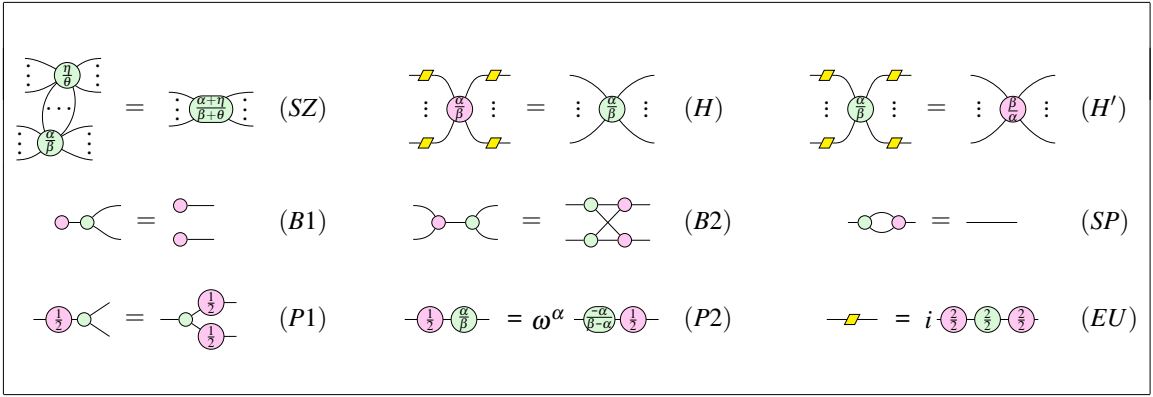

    \centering
    \scalebox{0.9}{$
    \def\arraystretch{3}
    \begin{array}{|clclcl|}
    \hline
    &&&&& \\[-2em]
    \input{XD_rules/spidernew.tikz}&(SZ)\qquad&\input{XD_rules/h.tikz}&(H)\qquad&\input{XD_rules/hp.tikz}&(H')\\
    \input{XD_rules/b1.tikz}&(B1)\qquad&\input{XD_rules/b2.tikz}&(B2)\qquad&\input{XD_rules/p.tikz}&(SP)\\
    \input{XD_rules/k1.tikz}&(P1)\qquad&\input{XD_rules/k2.tikz}&(P2)\qquad&\input{XD_rules/eu.tikz}&(EU)\\[1.5em]
    \hline
    \end{array}$
    }
    \caption{Rules for the flexsymmetric qutrit ZX-calculus. These hold for all $\alpha,\beta,\eta,\theta\in\R$, and for any permutation of the input and output wires. Additional useful derived rules are presented in Figure~\ref{fig:qutritphaseexactflexsymmetricderivedrules}.
    The letters stand respectively for (S)pider-Z, (H)adamard, (B)ialgebra, (SP)ecial, (P)auli, and (EU)ler decomposition.
    }\label{fig:qutritphaseexactflexsymmetricrules}
\end{figure}
Using these rules, other useful qutrit ZX-calculus rewrite rules may be derived. 
In particular, we can use these rules to prove the derived rules presented in Figure~\ref{fig:qutritphaseexactflexsymmetricderivedrules}.
\begin{figure}[!tb]
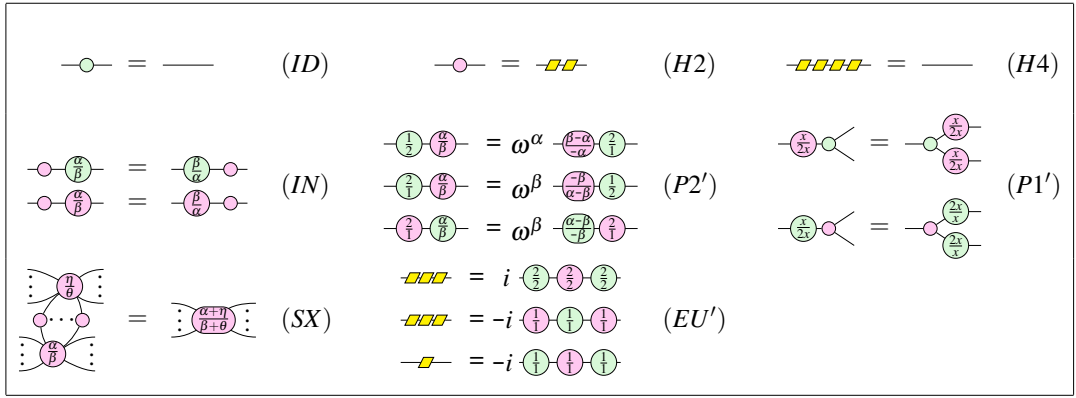

    \centering
    \scalebox{0.9}{$
    \def\arraystretch{3}
    \begin{array}{|clclcl|}
    \hline
    \input{XD_rules/s2.tikz}&(ID)\quad            &\input{XD_rules/h2.tikz}&(H2)\quad&\input{XD_rules/h4.tikz}&(H4)\\[0.5em]
    \input{XD_rules/dualizer_push.tikz}&(IN)\quad  &\input{XD_rules/k2p.tikz}&(P2')\quad&\input{XD_rules/k1p.tikz}&(P1')\\[1em]
    \input{XD_rules/spidernewprime.tikz}&(SX)\quad&\input{XD_rules/eup.tikz}&(EU')\quad&&\\[1.5em]
    \hline
    \end{array}$
    }
    \caption{These rules are derivable from the rules of Figure~\ref{fig:qutritphaseexactflexsymmetricrules} for any $\alpha,\beta,\eta,\theta\in\R$ and $x\in \{0,1,2\}$. The new letters stand respectively for (ID)entity, (IN)vert and (S)pider-X.}\label{fig:qutritphaseexactflexsymmetricderivedrules}
\end{figure}
As these rules are (a slight variation) on the non-flexsymmetric qutrit rules of Ref.~\cite{WangQ2018qutrit}, our calculus is also complete for the qutrit Clifford fragment (when ignoring non-zero scalars). 
The proofs of the derived rules of Figure~\ref{fig:qutritphaseexactflexsymmetricderivedrules} are given in Appendix~\ref{appendix:derivedrules}. We show in Appendix~\ref{app:necessity-of-rules} that most rules of Figure~\ref{fig:qutritphaseexactflexsymmetricrules} are necessary (i.e.~not derivable from the others).

We see in these rules that there is a special role for phases of the form $\frac{x}{2x}$ where $x\in \{0,1,2\}$. This is because $\begin{tikzpicture}
	\begin{pgfonlayer}{nodelayer}
		\node [style=none] (1) at (1, 0) {};
		\node [style=XD phase dot] (14) at (0, 0) {$\frac{x}{2x}$};
	\end{pgfonlayer}
	\begin{pgfonlayer}{edgelayer}
		\draw [style=none] (14) to (1.center);
	\end{pgfonlayer}
\end{tikzpicture}
 \propto \ket{x}$ and $\begin{tikzpicture}
	\begin{pgfonlayer}{nodelayer}
		\node [style=none] (1) at (1, 0) {};
		\node [style=Z phase dot] (14) at (0, 0) {$\frac{x}{2x}$};
		\node [style=none] (15) at (-1, 0) {};
	\end{pgfonlayer}
	\begin{pgfonlayer}{edgelayer}
		\draw [style=none] (14) to (1.center);
		\draw (15.center) to (14);
	\end{pgfonlayer}
\end{tikzpicture}
 = Z^x$. These relations can be derived by using the identity $1+\omega+\bar{\omega} = 0$ together with $\omega^2 = \omega^{-1} = \bar{\omega}$. In general we will see a lot of $\begin{tikzpicture}
	\begin{pgfonlayer}{nodelayer}
		\node [style=none] (1) at (1, 0) {};
		\node [style=Z phase dot] (14) at (0, 0) {$\frac{\alpha}{2\alpha}$};
		\node [style=none] (15) at (-1, 0) {};
	\end{pgfonlayer}
	\begin{pgfonlayer}{edgelayer}
		\draw [style=none] (14) to (1.center);
		\draw (15.center) to (14);
	\end{pgfonlayer}
\end{tikzpicture}
$ phases because they implement the $\ket{x}\mapsto \omega^{\alpha x}\ket{x}$ phase gate.
Additionally, note that the $(P2)$ rule on the qubit subspace is exactly the familiar qubit ZX rule \input{XD_rules/qubitk2.tikz}, since the red $\pi$ is the qubit Pauli $X$ while the pink $\frac{1}{2}$ phase is the qutrit Pauli $X$.

\section{Diagonal qutrit gates}\label{sec:diagonal-gates}

\subsection{Phase gadgets}
For qubits the concept of a \emph{phase gadget} has proven very useful. There's several different ways we can define a qubit phase gadget.
One way is to consider it as the diagonal gate $\ket{x,y}\mapsto e^{i\alpha (x\oplus y)} \ket{x,y}$ (for simplicity we are only considering a two-qubit phase gadget). This applies a phase of $e^{i\alpha}$ when $x\oplus y = 1$. Here $\oplus$ is the XOR operation, which is the addition on $\mathbb{Z}_2$. This suggests that we should define the qutrit phase gadget as $\ket{x,y}\mapsto e^{i\alpha (x+y)}\ket{x,y}$ where now we take $x+y$ to be modulo 3.

We could also define a phase gadget by its circuit realisation or diagrammatic representation. For qubits~\cite{KissingerA2020reduc}:
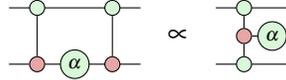
\[\input{phase-gadget-qubit.tikz}\]
We claim the qutrit variant of this construction is given by the following circuit which can be simplified to a similar diagrammatic representation:
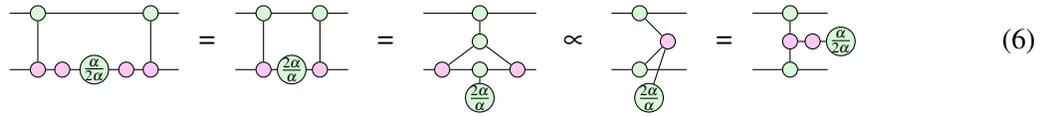
\begin{equation}\label{eq:phase-gadget-qutrit-simp}
    \input{phase-gadget-qutrit-simp.tikz}
\end{equation}
Indeed, inputting $\ket{x,y}$ into this diagram allows us to calculate its action:
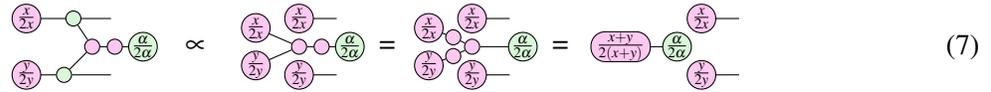
\begin{equation}\label{eq:phase-gadgets-qutrit-calc}
    \input{phase-gadget-qutrit-calc.tikz}
\end{equation}
This `floating scalar' expression evaluates to $\sqrt{3} \omega^{\alpha (x+y \text{ mod }3)}$, so that this diagram indeed implements the operation we want, and we see that these three ways to define a qubit phase gadget---via the action, via the circuit, or via the diagrammatic representation---are also equal for qutrits.

We can easily generalise this construction to an arbitrary number of qutrits:
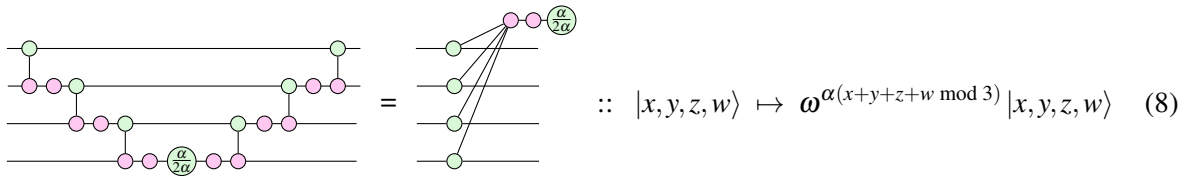
\begin{equation}\label{eq:phase-gadget-qutrit-multi}
    \input{phase-gadget-qutrit-multi.tikz} \ \ :: \ \ \ket{x,y,z,w} \ \mapsto \ \omega^{\alpha (x+y+z+w \text{ mod }3)} \ket{x,y,z,w} 
\end{equation}

We can also define more general phase gadgets where the phases don't have to be related to each other, i.e. we can replace $Z(\alpha, 2\alpha)$ with $Z(\alpha, \beta)$.
In this case we would still be calculating the value $x+y+z+w$ modulo 3, but then we apply a different phase depending on the value of this sum: if it is $0$ we don't apply any phase; if it is $1$ we apply $\omega^{\alpha}$; and if it is $2$ we apply $\omega^{\beta}$.

A particularly relevant choice of phases here is when $\alpha=\beta$. In this case, we apply the phase iff the sum value is not $0$. For a trit $x$ it turns out that $x^2 = 0$ if $x=0$ and $x^2 = 1$ otherwise --- this is actually a consequence of Fermat's little theorem and generalises to $x^{p-1} = 1$ modulo $p$ when $x\neq 0$ for $p$ prime.  Hence:
\begin{equation}\label{eq:Z-alpha-alpha}
    \begin{tikzpicture}
	\begin{pgfonlayer}{nodelayer}
		\node [style=Z phase dot] (0) at (0, 0) {$\frac{\alpha}{\alpha}$};
		\node [style=none] (1) at (-1, 0) {};
		\node [style=none] (2) at (1, 0) {};
	\end{pgfonlayer}
	\begin{pgfonlayer}{edgelayer}
		\draw (1.center) to (0);
		\draw (0) to (2.center);
	\end{pgfonlayer}
\end{tikzpicture}
 :: \ket{x} \mapsto \omega^{\alpha (x^2 \text{ mod }3)} \ket{x}.
\end{equation}

There is a complication with the phase gadget circuit representation that doesn't arise in the qubit setting, which is that the CNOT gate is self-inverse while the CX qutrit gate is not. In Eq.~\eqref{eq:phase-gadget-qutrit-simp} we needed to pair a CX with a CX$^\dagger$ to make the construction work. If we instead have a pair of CX gates, we get something a bit more complicated:
\begin{equation}\label{eq:phase-gadget-wrong-cnot}
    \input{phase-gadget-wrong-cnot.tikz}
\end{equation}

\begin{remark}
    Another way to view qubit phase gadgets is as an exponentiated Pauli $e^{i\alpha Z\otimes Z}$~\cite{phaseGadgetSynth,cowtan2020generic}. This however does \emph{not} generalise to qutrits, as the qutrit Pauli $Z$ is not self-adjoint, and hence cannot be exponentiated to give a unitary. In fact, a qutrit phase gadget cannot be represented as the exponential of a `pure tensor' like $e^{i\alpha A\otimes B}$. This does suggest that there could be another suitable generalisation of a phase gadget that is the exponential of a tensor of \emph{Gell-Mann matrices}, a qutrit basis of self-adjoint matrices.
\end{remark}

\subsection{Controlled phase gates}

The other type of useful diagonal gate for qubits is the \emph{controlled phase gate}. Such a gate applies a $Z(\alpha)$ gate on a qubit controlled on the value of a control.
There are multiple ways in which we can generalise these to the qutrit setting. The type of control we will consider first is the $\ket{2}$-control of Definition~\ref{def:ket2-controlled}.
To see how we can build a $\ket{2}$-controlled $Z$ phase gate, we will take inspiration from the qubit construction. Recall that there we have:
\begin{equation}\label{eq:controlled-Z-gate-decomp}
\input{controlled-Z-gate.tikz} \ = \ \input{controlled-Z-decomp.tikz}
\end{equation}
We can `port' the right-hand side to the qutrit setting, by taking each of the phases to be a $Z(\alpha,\beta)$. However, we then run into some problems. It is easy to check that when the top qutrit (the control) is $\ket{0}$ that the diagram indeed acts as the identity on the bottom qutrit (the target). However, it implements a different phase gate on the target depending on whether the control is in $\ket{1}$ or $\ket{2}$: 
\begin{equation}\label{eq:controlled-Z-qutrit}
    \input{controlled-Z-qutrit.tikz} \ \rightsquigarrow \ 
    \begin{cases}
        Z(0,0) \ &\text{if control is}\ \ket{0} \\
        Z(2\alpha-\beta,\alpha+\beta) \ &\text{if control is}\ \ket{1} \\
        Z(\alpha+\beta,2\beta-\alpha) \ &\text{if control is}\ \ket{2}
    \end{cases}
\end{equation}
Seeing as we want to construct the $\ket{2}$-controlled gate that should act as the identity when the control is $\ket{1}$ this is a problem.
We solve this issue by `doubling up' the construction, with the second construction being conjugated by $X_{12}$ on the control in order to interchange the role of $\ket{1}$ and $\ket{2}$:
\begin{equation}\label{eq:controlled-Z-qutrit-paired}
    \input{controlled-Z-qutrit-paired.tikz}
\end{equation}
By referring to Eq.~\eqref{eq:controlled-Z-qutrit} we see then that in order for Eq.~\eqref{eq:controlled-Z-qutrit-paired} to be equal to the $\ket{2}$-controlled $Z(\theta,\phi)$ gate it needs to satisfy a set of linear equations.
We can solve these to get a (unique up to some Clifford phases) solution:
\begin{equation}
    \alpha \ =\  \frac{\theta - \phi}{3} \quad\qquad \beta \ =\  \frac{\theta}{3} \quad\qquad \gamma \ =\  \frac{\phi}{3} \quad\qquad \delta \ =\  \frac{\phi - \theta}{3}
\end{equation}
We can hence write any $\ket{2}$-controlled phase gate using at most four CX gates and four phase gates. For example, if we pick $\theta = 1$ and $\phi = 2$ (so that we are constructing the controlled $Z$ gate) we get:
\begin{equation}\label{eq:ket2-Z}
    \input{ket2-Z-decomp.tikz}
\end{equation}

Here we write this blue dot with a $2$ in it to denote a $\ket{2}$-control.
We see then that our construction in the special case of $\ket{2}$-controlled Paulis indeed achieves the lowest known $T$-count of $3$~\cite{BocharovA2016ternaryarithmetics}.  By conjugating the control wire by either $X_{+1}$ or $X_{-1}$ we can make the gate instead be controlled on either $\ket{1}$ or $\ket{0}$.
 
We can add any number of controls to our construction in Eq.~\eqref{eq:controlled-Z-qutrit-paired} to make it controlled on any tritstring.  Without loss of generality, let us say the tritstring in question is $\ket{2}^{\otimes n}$ (by conjugating with $X$ gates, we can make this into a control on any tritstring of length $n$).  
The na\"ive way to construct this controlled gate is to inductively add controls to each gate in the decomposition: controlled constructions for the $X$ or CX gate are described in, for instance, Ref.~\cite{YehL2022qutritcontrolledcliffordplust}, while controlled $Z$ phase gates can be constructed by recursively applying Eq.~\eqref{eq:controlled-Z-qutrit-paired}.  However, this method is not efficient as it requires an exponential number of gates as the number of controls increases.

We can do better by not adding controls to all the gates in the decomposition:
\begin{equation}\label{eq:tritstring-controlled-Z}
    \input{tritstring-controlled-Z.tikz}
\end{equation}
In the case where all the controlled gates fire, this indeed implements any desired $Z$ phase gate on the target qutrit.  Otherwise, none of the controlled gates fire, and then the bottom two qutrits becomes identity (use $(H2)$ and $(H4)$ on the top qutrit and $(SZ)$ and $(ID)$ on the bottom qutrit).
We hence get the following proposition.

\begin{proposition}
    Any tritstring-controlled qutrit $Z$ or $X$ phase gate can be constructed without ancillae and with a polynomial number of Clifford and phase gates.
\end{proposition}
\begin{proof}
    The $X$ phase gates can be constructed from the $Z$ phase gates by conjugating by Hadamards, so we only need to describe how to construct tritstring-controlled $Z$ phase gates.
    Suppose we wish to construct a phase gate with $n$ controls.  By our prior work in Ref.~\cite{YehL2022qutritcontrolledcliffordplust}, each $\ket{2}^{\otimes (n-1)}$-controlled CX gate can be built ancilla-free using $O(n^{2.585})$ qutrit Clifford+$T$ gates.
    It then remains to show how to construct the $\ket{2}^{\otimes (n-1)}$-controlled $Z$ phase gate in Eq.~\eqref{eq:tritstring-controlled-Z}. We do this recursively. 
    To construct the gate with $k$ controls, we need four controlled CX gates with $k-1$ controls and a $k-1$ controlled $Z$ phase gate, which then needs four controlled CX gates with $k-2$ controls, and so on.  The total asymptotic gate count is then $4 O(n^{2.585}) + 4 O((n-1)^{2.585}) + ...$ which gives us a gate count of $O(n^{3.585})$.
\end{proof}
Note that in this construction, the size of the phases involved becomes exponentially smaller in the number of controls.
We will next see that there is a notion of control which circumvents this issue.

\subsection{Phase multipliers}
The $\ket{2}$-controlled phase gate is just one possible way to extend the idea of a controlled-phase gate from qubits. Another way is to realise that for qubits we can describe the action of a controlled phase gate as $\ket{x,y} \mapsto e^{i\alpha \,x\cdot y} \ket{x,y}$. Indeed, if the control qubit is in the state $x=0$, then this is just the identity, while if $x=1$, we apply $e^{i\alpha y}$ which corresponds to a $Z(\alpha)$ gate on the $\ket{y}$ qubit.

We see then that while a phase gadget is based on the addition operation of $\mathbb{Z}_2$, controlled phase gates are based on the multiplication operation of $\mathbb{Z}_2$. This suggests that the controlled phase gate equivalent for qutrits should be $\ket{x,y}\mapsto e^{i\alpha x\cdot y}\ket{x,y}$ where now we take $x\cdot y$ modulo 3. We will show how we can construct this operation using phase gadgets. In order to distinguish this type of gate from the previously considered controlled phase gates, we will refer to a gate where the phase depends on $x\cdot y$ as a \emph{phase multiplier}. Before we show how to build phase multipliers for qutrits, we first need to understand how to build them for qubits.
For bits $x$ and $y$ we have the relation 
\begin{equation}\label{eq:bit-plus-rel}
x\cdot y = \frac12 (x+y- (x\oplus y)).
\end{equation}
Importantly, we are considering the~$+$ operation here not modulo $2$, but just as an action on real numbers, and we are writing $\oplus$ for addition modulo $2$. Using this relation we can write $e^{i\alpha (x\cdot y)} = e^{i\frac12\alpha (x+y-(x\oplus y))} = e^{i\frac12 \alpha x} e^{i\frac12\alpha y} e^{-i\frac12\alpha (x\oplus y)}$. This is where the circuit decomposition of Eq.~\eqref{eq:controlled-Z-gate-decomp} comes from.
This relation between additive and multiplicative phase gates follows from a Fourier-type duality that exists for semi-Boolean functions, which is explored in detail in Ref.~\cite{GraphicalFourier2019}.

It turns out that a similar decomposition is possible for qutrits. Note that we can derive Eq.~\eqref{eq:bit-plus-rel} by starting with the equation $(x+y)^2 = x^2 + y^2 + 2x\cdot y$ and then realising that $x^2 = x$ for $x\in \{0,1\}$ so that this reduces to $x\oplus y = x + y + 2x\cdot y$ for bits. When working with trits we can't remove these squares, but we can still get a useful relation. Bring terms to the other side to get $-2x\cdot y = x^2 + y^2 - (x+y)^2$ and then use the fact that modulo 3 we have $-2 = 1$ to get $x\cdot y = x^2 + y^2 - (x+y)^2$. It is now straightforward to check that this continues to hold when we interpret the outer $+$ and $-$ here not modulo 3, but as operations on the real numbers, so that we get the relation:
\begin{equation}\label{eq:trit-plus-rel}
    x\cdot y~\text{mod } 3 = (x^2~\text{mod } 3) + (y^2~\text{mod } 3) - ((x+y)^2~\text{mod } 3)
\end{equation}
Hence, using Eq.~\eqref{eq:Z-alpha-alpha} we get the following decomposition:
\begin{equation}\label{eq:controlled-qutrit-phase}
    \input{controlled-qutrit-phase-decomp.tikz} \qquad :: \qquad \ket{x,y} \mapsto \omega^{\alpha (x\cdot y\text{ mod }3)} \ket{x,y}
\end{equation}
We can easily generalise Eq.~\eqref{eq:trit-plus-rel} to as many variables as desired by iterating it.  For three trits:
\begin{align}
    (x\cdot y)\cdot z\  &=\  x^2\cdot z + y^2\cdot z - (x+y)^2\cdot z \nonumber \\
    &=\  x^4 + z^2 - (x^2 + z)^2 + y^4 + z^2 - (y^2 + z)^2 - (x+y)^4 - z^2 + ((x+y)^2+z)^2 \nonumber \\
    &=\  x^2 + y^2 + z^2 - (x^2 + z)^2 - (y^2 + z)^2 - (x+y)^2 + ((x+y)^2+z)^2 \label{eq:triple-product}
\end{align}
Here we used that $x^4 = x^2$ modulo 3.

Note that Eq.~\eqref{eq:Z-alpha-alpha} shows how to apply a phase proportional to the input trit squared modulo $3$.  However, in order to use this trick to apply a phase proportional to a higher order term such as $(y+x^2)^2$, we need a way to compute $y+x^2$ and store it ``on the wire''.  In other words, we need to construct a circuit for the unitary defined by $\ket{x,y}\mapsto \ket{x,y+x^2}$.  Because this simply adds $1$ (modulo 3) to $y$ iff $x\neq 0$, we construct it by adding $1$ to the second qubit, and then applying a $\ket{0}$-controlled $X_{-1}$ gate.  To build this gate, we use the $\ket{2}$-controlled $Z$ gate that we built from two phase gadgets in Eq.~\eqref{eq:ket2-Z}:
\begin{equation}\label{eq:square-control}
    \input{square-control.tikz} \quad :: \quad \ket{x,y} \mapsto \ket{x,y+x^2}
\end{equation}
The above trick was also described in Ref.~\cite{BocharovA2016ternaryarithmetics}.  We further use this to build the type of phase gate below.
\begin{equation}\label{eq:square-phase}
    \input{square-phase.tikz} \quad :: \quad \ket{x,y} \mapsto \omega^{\alpha (y+x^2)^2} \ket{x,y}.
\end{equation}
We now have all the ingredients necessary to build the phase multiplier corresponding to the formula~\eqref{eq:triple-product}.
What is interesting about this is that we do not have to use smaller factors of $\alpha$. This is in contrast to the qubit counterpart of the formula~\eqref{eq:triple-product} where we get a factor of $\frac{1}{4}$, due to the factor of $\frac{1}{2}$ in Eq.~\eqref{eq:bit-plus-rel}.  In fact, for qubits, the generalisation to $n$ variables will have a prefactor of $(1/2)^{n-1}$ so that for instance the three-qubit-controlled $Z$ and controlled $T$ gates cannot be constructed without ancillae in Clifford+$T$~\cite{GilesB2013multiqubitcliffordplustsynthesis}, as we need $\pi/8$ phase gates. Instead, no matter the number of qutrits, we do not get such a prefactor and can iteratively construct it as in the formula~\eqref{eq:triple-product}, as we did for qutrit Clifford+$T$ in Ref.~\cite{YehL2022qutritcontrolledcliffordplust}.
The circuit~\eqref{eq:square-phase}, alongside the square phase of Eq.~\eqref{eq:Z-alpha-alpha} suffices to generalise Eq.~\eqref{eq:controlled-qutrit-phase} to any number of qutrits.
\begin{proposition}\label{prop:phasemultiplier}
    We can construct, without ancillae and using $O(2^n)$ Clifford+$T$, $Z(\alpha,\alpha)$, and $Z(-\alpha,-\alpha)$ gates, the $n$-qutrit phase multiplier gate defined by $\ket{x_1,\ldots, x_n} \mapsto \omega^{\alpha \left(\left(x_1\cdots x_n\right) \text{ mod } 3\right)} \ket{x_1,\ldots, x_n}$.
\end{proposition}
\noindent See Appendix~\ref{app:phase-multipliers} for the details.



\begin{remark}
    In Ref.~\cite{CuiS2017Diagonalhierarchy} the diagonal gates at all levels of the Clifford hierarchy are analysed for any qudit of prime dimension. They show for instance that the gate implementing $\ket{x_1\cdots x_n}\mapsto \omega^{x_1\cdots x_n} \ket{x_1\cdots x_n}$ (which is the $n$-controlled $2\pi/3$ phase multiplier gate) is in the $n$th level of the Clifford hierarchy. This might be surprising as our construction shows how to build this gate, for any $n$, only using gates from the third level of the Clifford hierarchy (namely Clifford gates and the $T$ gate). However, note that while the \emph{diagonal} gates on a level of the hierarchy form a group, the full set of (not necessarily diagonal) gates is not closed under composition, and hence we can build higher-level unitaries using lower-level ones.
\end{remark}

\section{Applications}\label{sec:emulation-application}
We've now seen that we can use phase gadgets to build a number of useful diagonal unitaries. In this section we will see how we can build more general diagonal qutrit unitaries, and specifically those that \emph{emulate} qubit operations.  
Qudit emulation of qubit operations can result in efficiency gains, by using higher level states rather than ancillae.  While there has been significant work on emulating qubits using qutrits and qudits, much of this has been limited to realising gates within classical reversible computing such as multiple-controlled Toffolis.  
In contrast, fewer works have addressed qutrit gate sets containing arbitrary phases.  Examples include a $\ket{2}$-controlled $Z(0,\phi)$ decomposition in terms of qutrit-controlled qubit $\phi/3$ rotations~\cite{DiY2013synthesis} or quantum multiplexers and uniformly-controlled Givens rotations from the cosine-sine decomposition~\cite{KhanF2006synthesisqudit}.

Throughout this section, we will write $\,\stackrel{\mathclap{\normalfont\mbox{e}}}{=}\,$ to denote that a qubit unitary is emulated by a qutrit unitary.

We will first see how to emulate arbitrary qubit diagonal unitaries.
Note that when we restrict to the $\{\ket{0},\ket{1}\}$ subspace, that a qutrit phase multiplier $\ket{x_1,\ldots, x_n} \mapsto e^{i\alpha \left(\left(x_1\cdots x_n\right) \text{ mod } 3\right)} \ket{x_1,\ldots, x_n}$ only applies a phase $\alpha$ if and only if all $n$ qubits are in the $\ket{1}$ state.
Hence, for instance, the two-qubit C$Z(\alpha)$ gate is directly emulated by its two-qutrit counterpart of Eq.~\eqref{eq:controlled-qutrit-phase} with action $\ket{x,y} \mapsto e^{i\alpha x\cdot y} \ket{x,y}$. Consequently, by Proposition~\ref{prop:phasemultiplier} we see that using qutrit Clifford+$T$ gates along with $Z(\alpha,\alpha)$ and $Z(-\alpha,-\alpha)$ we can emulate the multiple-controlled $Z(\alpha)$ qubit gate without ancillae.

Now, by conjugating a multiple-controlled C$Z(\alpha)$ gate by the appropriate $X_{01}$ gates, we can decide on which input the $e^{i\alpha}$ phase is applied. Using multiple of these gates we can then arbitrarily decide for each input which phase should be applied to it. This then allows us to emulate an arbitrary diagonal qubit gate.
\begin{proposition}
    We can emulate the $\text{diag}(\omega^{\alpha_1},\ldots, \omega^{\alpha_{2^n}})$ qubit unitary using qutrit Clifford+$T$, $Z(\alpha_j,\alpha_j)$ and $Z(-\alpha_j,-\alpha_j)$ gates and without using ancillae.
\end{proposition}
When using a standard qubit unitary synthesis algorithm, the desired phases $e^{i\alpha_j}$ would be implemented using many-controlled phase gates that require exponentially small angles $e^{i\alpha_j/2^n}$, which is problematic when the use-case is in fault-tolerant computing where non-Clifford phase gates must be constructed using magic state distillation and injection.
Our construction could hence lead to some benefits in synthesising diagonal qubit unitaries using less non-Clifford resources.

It turns out that for the specific case of a qubit CCZ gate, that we can emulate it using qutrits in an even more efficient way.
While we could use the emulation construction above, it turns out to be better to consider an altered construction.

\begin{lemma}\label{lemma:qubitccu}
    Given a qutrit $\ket{2}$-controlled $U$ gate for an emulated qubit unitary $U$, we can construct a qutrit emulation of the qubit CCU gate with the same non-Clifford cost as the $\ket{2}$-controlled $U$ gate.
\end{lemma}
\begin{proof}
    One can readily verify, by initializing the top two qubits to $\{\ket{00},\ket{01},\ket{10},\text{and}\ket{11}\}$, that the below qutrit decomposition from Ref.~\cite{GokhaleP2019asymptotic} emulates the qubit CCU gate.
    \begin{equation}\label{eq:qubitccu}
        \input{qubitccuemu.tikz}
    \end{equation}
    We can then replace the two non-Clifford $\ket{1}$-controlled $X_{+1}$ and $\ket{1}$-controlled $X_{-1}$ by CX and CX$^\dagger$, preserving correctness of the emulation as the action on the $\{\ket{0},\ket{1}\}$ subspace is unchanged~\cite{BocharovA2017ternaryshor}.
\end{proof}

Using this lemma we see that to get an efficient emulation of the qubit CCZ, it remains to find an efficient qutrit emulation of the $\ket{2}$-controlled qubit $Z$ gate.
\begin{lemma}
    Let $U=\text{diag}(1,\omega^{\eta})$ be an arbitrary qubit $Z$ phase gate.
    Then we can build the $\ket{2}$-controlled emulated $U$ using the controlled phase gate of Eq.~\eqref{eq:controlled-Z-qutrit}:
    \begin{equation}
        \input{phase-gadget-emu.tikz}
    \end{equation}
    Here $\alpha$ and $\beta$ satisfy, for some $k \in \mathbb{Z}$, $2\alpha - \beta = 3 k$ and $\alpha + \beta = \eta$
    and the questionmarks $?$ denote that these phases are irrelevant for the emulation.
\end{lemma}

If we choose $\alpha = 3/2$ and $\beta = 0$ in this construction we are emulating the $\ket{2}$-controlled qubit $Z$ gate, because the phase of $\omega^{3/2} = -1$ applies iff the control qutrit is $\ket{2}$ and the target qubit is $\ket{1}$.
Note that a $Z(3/2,0)$ phase is equal to $X_{12} \ R \ X_{12}$, referring to the $R$ gate from Definition~\ref{def:Rgate}. Hence:
\begin{corollary}\label{cor:tcqubitz}
    The $\ket{2}$-controlled qubit $Z$ gate can be emulated with $R$-count $3$.
\end{corollary}

Combining this corollary with Lemma~\ref{lemma:qubitccu} we arrive at our result.
\begin{proposition}
    The qubit CCZ gate can be emulated ancilla-free in qutrit Clifford+$R$ with $R$-count 3.
\end{proposition}

As shown in Ref.~\cite{HowardM2017resourcetheorymagic}, any implementation of a CCZ gate requires at least four qubit $T$ gates. Additionally any unitary implementation requires at least seven~\cite{GossetD2014opttoffolit}. Here we see that surprisingly, by embedding the CCZ into qutrit space, we can construct it using just three non-Clifford single-qutrit gates and that moreover this is unitary and ancilla-free.  This is also a new minimum amongst qudit emulations: for instance, in Ref.~\cite{HeyfronL2019quditcompiler}, they needed four qudit (for prime $d > 3$) $T$ gates to emulate qubit CCZ.

\section{Conclusion}
\label{sec:conclusion}
We introduced phase gadgets in the qutrit ZX-calculus.
To do this, we adapted the original qutrit ZX-calculus to be flexsymmetric so that the phase gadgets' behaviour would not depend on the directionality of their edges.
Using phase gadgets we showed how to build two types of qutrit controlled phase gates: tritstring-controlled phase gates and phase multipliers. This allowed us to emulate the qubit CCZ gate using just three single-qudit non-Clifford gates.

While some of our constructions will naturally generalise to arbitrary qudit dimension, some things are qutrit specific.
It seems to be a coincidence that for qutrits, in contrast with other-dimensional qudits, you can derive a relation between modular multiplication and addition~\eqref{eq:trit-plus-rel} from the same binomial as for qubits~\eqref{eq:bit-plus-rel}, which comes from having a natural way to express $x^2$ mod 3 thanks to Fermat's little theorem.  As a result, qubit and qutrit phase multipliers admit constructions which are structurally similar, despite the fact that for qubits it applies a phase of $\alpha$ on only one possible input --- where all $n$ qubits are $\ket{1}$ --- while for qutrits it applies a phase, which can be $\alpha$ or $2\alpha$, for $2^n$ of the $3^n$ possible input basis states.  Moreover, it seems quite special that Eq.~\eqref{eq:trit-plus-rel} does not have any factors making the size of the phases internal to the decomposition decrease (in contrast to the qubit case).

We believe we could use these results as a stepping stone towards defining a qutrit \emph{ZH-calculus}~\cite{EPTCS287.2}.
In the qubit ZH-calculus, the H-boxes represent matrices with coefficients $a^{i_1 ... i_m j_1 ... j_n}$ for a complex number~$a$ and $i_1,...,i_m,j_1,...,j_n \in \{0,1\}$.  Therefore, the obvious generalisation to qutrits (at least for $a$ a complex phase) corresponds to our qutrit phase multipliers. Phase gadgets and phase multipliers could then be related in the same way as they are for qubit ZX and ZH~\cite{GraphicalFourier2019}.

An open problem is to find a suitable qutrit equivalent of exponentiated Paulis.  The canonical self-adjoint generalisation of qubit Paulis to qutrits, the Gell-Mann matrices, can be exponentiated to unitaries, but it is not clear how they are related to the qutrit Paulis exactly. A starting point to find the proper relation here is to express exponentiations using a Hermitian operator basis constructed from the qutrit Paulis~\cite{AsadianA2015hermitianhwops}.

Finally, let us mention that based on work on an earlier draft of this paper, a proposed scheme for physically implementing a qutrit phase gadget in superconducting qutrit hardware was made~\cite{CaoS2022qutritzxsuperconducting}.

\textbf{Acknowledgements}: JvdW is supported by an NWO Rubicon personal fellowship. LY is supported by an Oxford - Basil Reeve Graduate Scholarship at Oriel College with the Clarendon Fund. The authors wish to thank Aleks Kissinger, Shuxiang Cao, Alex Cowtan and Will Simmons for valuable discussions.

\bibliographystyle{eptcs}
\bibliography{main}

\begin{thebibliography}{10}
\providecommand{\bibitemdeclare}[2]{}
\providecommand{\surnamestart}{}
\providecommand{\surnameend}{}
\providecommand{\urlprefix}{Available at }
\providecommand{\url}[1]{\texttt{#1}}
\providecommand{\href}[2]{\texttt{#2}}
\providecommand{\urlalt}[2]{\href{#1}{#2}}
\providecommand{\doi}[1]{doi:\urlalt{http://dx.doi.org/#1}{#1}}
\providecommand{\bibinfo}[2]{#2}

\bibitemdeclare{inproceedings}{AmyVerification}
\bibitem{AmyVerification}
\bibinfo{author}{Matthew \surnamestart Amy\surnameend} (\bibinfo{year}{2019}):
  \emph{\bibinfo{title}{{Towards Large-scale Functional Verification of
  Universal Quantum Circuits}}}.
\newblock In \bibinfo{editor}{Peter \surnamestart Selinger\surnameend} \&
  \bibinfo{editor}{Giulio \surnamestart Chiribella\surnameend}, editors: {\sl
  \bibinfo{booktitle}{{Proceedings of the 15th International Conference on}
  Quantum Physics and Logic, {Halifax, Canada, 3-7th June 2018}}}, {\sl
  \bibinfo{series}{Electronic Proceedings in Theoretical Computer Science}}
  \bibinfo{volume}{287}, \bibinfo{publisher}{Open Publishing Association}, pp.
  \bibinfo{pages}{1--21}, \doi{10.4204/EPTCS.287.1}.

\bibitemdeclare{article}{AnwarH2012r2distillation}
\bibitem{AnwarH2012r2distillation}
\bibinfo{author}{Hussain \surnamestart Anwar\surnameend},
  \bibinfo{author}{Earl~T \surnamestart Campbell\surnameend} \&
  \bibinfo{author}{Dan~E \surnamestart Browne\surnameend}
  (\bibinfo{year}{2012}): \emph{\bibinfo{title}{Qutrit magic state
  distillation}}.
\newblock {\sl \bibinfo{journal}{New Journal of Physics}}
  \bibinfo{volume}{14}(\bibinfo{number}{6}), p. \bibinfo{pages}{063006},
  \doi{10.1088/1367-2630/14/6/063006}.

\bibitemdeclare{article}{AsadianA2015hermitianhwops}
\bibitem{AsadianA2015hermitianhwops}
\bibinfo{author}{Ali \surnamestart Asadian\surnameend}, \bibinfo{author}{Paul
  \surnamestart Erker\surnameend}, \bibinfo{author}{Marcus \surnamestart
  Huber\surnameend} \& \bibinfo{author}{Claude \surnamestart
  Kl\"ockl\surnameend} (\bibinfo{year}{2016}):
  \emph{\bibinfo{title}{{Heisenberg-Weyl Observables: Bloch vectors in phase
  space}}}.
\newblock {\sl \bibinfo{journal}{Phys. Rev. A}} \bibinfo{volume}{94}, p.
  \bibinfo{pages}{010301}, \doi{10.1103/PhysRevA.94.010301}.

\bibitemdeclare{inproceedings}{EPTCS287.2}
\bibitem{EPTCS287.2}
\bibinfo{author}{Miriam \surnamestart Backens\surnameend} \&
  \bibinfo{author}{Aleks \surnamestart Kissinger\surnameend}
  (\bibinfo{year}{2019}): \emph{\bibinfo{title}{{ZH: A Complete Graphical
  Calculus for Quantum Computations Involving Classical Non-linearity}}}.
\newblock In \bibinfo{editor}{Peter \surnamestart Selinger\surnameend} \&
  \bibinfo{editor}{Giulio \surnamestart Chiribella\surnameend}, editors: {\sl
  \bibinfo{booktitle}{Proceedings of the 15th International Conference on
  Quantum Physics and Logic, Halifax, Canada, 3-7th June 2018}}, {\sl
  \bibinfo{series}{Electronic Proceedings in Theoretical Computer Science}}
  \bibinfo{volume}{287}, \bibinfo{publisher}{Open Publishing Association}, pp.
  \bibinfo{pages}{23--42}, \doi{10.4204/EPTCS.287.2}.

\bibitemdeclare{article}{Backens2020extraction}
\bibitem{Backens2020extraction}
\bibinfo{author}{Miriam \surnamestart Backens\surnameend},
  \bibinfo{author}{Hector \surnamestart Miller-Bakewell\surnameend},
  \bibinfo{author}{Giovanni \surnamestart de~Felice\surnameend},
  \bibinfo{author}{Leo \surnamestart Lobski\surnameend} \&
  \bibinfo{author}{John \surnamestart van~de Wetering\surnameend}
  (\bibinfo{year}{2021}): \emph{\bibinfo{title}{{There and back again: A
  circuit extraction tale}}}.
\newblock {\sl \bibinfo{journal}{{Quantum}}} \bibinfo{volume}{5}, p.
  \bibinfo{pages}{421}, \doi{10.22331/q-2021-03-25-421}.

\bibitemdeclare{inproceedings}{deBeaudrapN2020treducspidernest}
\bibitem{deBeaudrapN2020treducspidernest}
\bibinfo{author}{Niel \surnamestart de~Beaudrap\surnameend},
  \bibinfo{author}{Xiaoning \surnamestart Bian\surnameend} \&
  \bibinfo{author}{Quanlong \surnamestart Wang\surnameend}
  (\bibinfo{year}{2020}): \emph{\bibinfo{title}{{Fast and Effective Techniques
  for T-Count Reduction via Spider Nest Identities}}}.
\newblock In \bibinfo{editor}{Steven~T. \surnamestart Flammia\surnameend},
  editor: {\sl \bibinfo{booktitle}{15th Conference on the Theory of Quantum
  Computation, Communication and Cryptography (TQC 2020)}}, {\sl
  \bibinfo{series}{Leibniz International Proceedings in Informatics (LIPIcs)}}
  \bibinfo{volume}{158}, \bibinfo{publisher}{Schloss Dagstuhl--Leibniz-Zentrum
  f{\"u}r Informatik}, \bibinfo{address}{Dagstuhl, Germany}, pp.
  \bibinfo{pages}{11:1--11:23}, \doi{10.4230/LIPIcs.TQC.2020.11}.

\bibitemdeclare{inproceedings}{deBeaudrapN2020reducepifourphase}
\bibitem{deBeaudrapN2020reducepifourphase}
\bibinfo{author}{Niel \surnamestart de~Beaudrap\surnameend},
  \bibinfo{author}{Xiaoning \surnamestart Bian\surnameend} \&
  \bibinfo{author}{Quanlong \surnamestart Wang\surnameend}
  (\bibinfo{year}{2020}): \emph{\bibinfo{title}{{Techniques to Reduce
  $\pi/4$-Parity-Phase Circuits, Motivated by the ZX Calculus}}}.
\newblock In \bibinfo{editor}{Bob \surnamestart Coecke\surnameend} \&
  \bibinfo{editor}{Matthew \surnamestart Leifer\surnameend}, editors: {\sl
  \bibinfo{booktitle}{Proceedings 16th International Conference on Quantum
  Physics and Logic, Chapman University, Orange, CA, USA., 10-14 June 2019}},
  {\sl \bibinfo{series}{Electronic Proceedings in Theoretical Computer
  Science}} \bibinfo{volume}{318}, \bibinfo{publisher}{Open Publishing
  Association}, pp. \bibinfo{pages}{131--149}, \doi{10.4204/EPTCS.318.9}.

\bibitemdeclare{article}{BlokM2021scrambling}
\bibitem{BlokM2021scrambling}
\bibinfo{author}{M.~S. \surnamestart Blok\surnameend}, \bibinfo{author}{V.~V.
  \surnamestart Ramasesh\surnameend}, \bibinfo{author}{T.~\surnamestart
  Schuster\surnameend}, \bibinfo{author}{K.~\surnamestart O'Brien\surnameend},
  \bibinfo{author}{J.~M. \surnamestart Kreikebaum\surnameend},
  \bibinfo{author}{D.~\surnamestart Dahlen\surnameend},
  \bibinfo{author}{A.~\surnamestart Morvan\surnameend},
  \bibinfo{author}{B.~\surnamestart Yoshida\surnameend}, \bibinfo{author}{N.~Y.
  \surnamestart Yao\surnameend} \& \bibinfo{author}{I.~\surnamestart
  Siddiqi\surnameend} (\bibinfo{year}{2021}): \emph{\bibinfo{title}{{Quantum
  Information Scrambling on a Superconducting Qutrit Processor}}}.
\newblock {\sl \bibinfo{journal}{Phys. Rev. X}} \bibinfo{volume}{11}, p.
  \bibinfo{pages}{021010}, \doi{10.1103/PhysRevX.11.021010}.

\bibitemdeclare{article}{BocharovA2016ternaryarithmetics}
\bibitem{BocharovA2016ternaryarithmetics}
\bibinfo{author}{Alex \surnamestart Bocharov\surnameend},
  \bibinfo{author}{Shawn \surnamestart Cui\surnameend}, \bibinfo{author}{Martin
  \surnamestart Roetteler\surnameend} \& \bibinfo{author}{Krysta \surnamestart
  Svore\surnameend} (\bibinfo{year}{2016}): \emph{\bibinfo{title}{{Improved
  Quantum Ternary Arithmetics}}}.
\newblock {\sl \bibinfo{journal}{Quantum Information and Computation}}
  \bibinfo{volume}{16}, pp. \bibinfo{pages}{862--884},
  \doi{10.26421/QIC16.9-10-8}.

\bibitemdeclare{article}{BocharovA2017ternaryshor}
\bibitem{BocharovA2017ternaryshor}
\bibinfo{author}{Alex \surnamestart Bocharov\surnameend},
  \bibinfo{author}{Martin \surnamestart Roetteler\surnameend} \&
  \bibinfo{author}{Krysta~M. \surnamestart Svore\surnameend}
  (\bibinfo{year}{2017}): \emph{\bibinfo{title}{{Factoring with qutrits: Shor's
  algorithm on ternary and metaplectic quantum architectures}}}.
\newblock {\sl \bibinfo{journal}{Phys. Rev. A}} \bibinfo{volume}{96}, p.
  \bibinfo{pages}{012306}, \doi{10.1103/PhysRevA.96.012306}.

\bibitemdeclare{article}{BullockS2005qubitemuqudit}
\bibitem{BullockS2005qubitemuqudit}
\bibinfo{author}{Stephen \surnamestart Bullock\surnameend},
  \bibinfo{author}{Dianne \surnamestart O'Leary\surnameend} \&
  \bibinfo{author}{Gavin \surnamestart Brennen\surnameend}
  (\bibinfo{year}{2005}): \emph{\bibinfo{title}{{Asymptotically Optimal Quantum
  Circuits for d-level Systems}}}.
\newblock {\sl \bibinfo{journal}{Physical Review Letters}}
  \bibinfo{volume}{94}(\bibinfo{number}{23}),
  \doi{10.1103/physrevlett.94.230502}.

\bibitemdeclare{article}{CampbellE2014quditmsdthresholds}
\bibitem{CampbellE2014quditmsdthresholds}
\bibinfo{author}{Earl~T. \surnamestart Campbell\surnameend}
  (\bibinfo{year}{2014}): \emph{\bibinfo{title}{{Enhanced Fault-Tolerant
  Quantum Computing in $d$-Level Systems}}}.
\newblock {\sl \bibinfo{journal}{Phys. Rev. Lett.}} \bibinfo{volume}{113}, p.
  \bibinfo{pages}{230501}, \doi{10.1103/PhysRevLett.113.230501}.

\bibitemdeclare{article}{CampbellE2012tgatedistillation}
\bibitem{CampbellE2012tgatedistillation}
\bibinfo{author}{Earl~T. \surnamestart Campbell\surnameend},
  \bibinfo{author}{Hussain \surnamestart Anwar\surnameend} \&
  \bibinfo{author}{Dan~E. \surnamestart Browne\surnameend}
  (\bibinfo{year}{2012}): \emph{\bibinfo{title}{{Magic-State Distillation in
  All Prime Dimensions Using Quantum Reed-Muller Codes}}}.
\newblock {\sl \bibinfo{journal}{Phys. Rev. X}} \bibinfo{volume}{2}, p.
  \bibinfo{pages}{041021}, \doi{10.1103/PhysRevX.2.041021}.

\bibitemdeclare{misc}{CaoS2022qutritzxsuperconducting}
\bibitem{CaoS2022qutritzxsuperconducting}
\bibinfo{author}{Shuxiang \surnamestart Cao\surnameend}, \bibinfo{author}{Lia
  \surnamestart Yeh\surnameend}, \bibinfo{author}{Bakr~Mustafa \surnamestart
  S\surnameend}, \bibinfo{author}{Giulio \surnamestart Campanaro\surnameend},
  \bibinfo{author}{Simone~D \surnamestart Fasciati\surnameend},
  \bibinfo{author}{James~F \surnamestart Wills\surnameend},
  \bibinfo{author}{Boris \surnamestart Shteynas\surnameend},
  \bibinfo{author}{Vivek \surnamestart Chidambaram\surnameend},
  \bibinfo{author}{John \surnamestart van~de Wetering\surnameend} \&
  \bibinfo{author}{Peter~J \surnamestart Leek\surnameend}
  (\bibinfo{year}{2022}): \emph{\bibinfo{title}{{Qutrit-ZX Calculus on
  superconducting transmon qutrits}}}.
\newblock \urlprefix\url{https://meetings.aps.org/Meeting/MAR22/Session/F37.7}.

\bibitemdeclare{article}{Carette2021OTM}
\bibitem{Carette2021OTM}
\bibinfo{author}{Titouan \surnamestart Carette\surnameend}
  (\bibinfo{year}{2021}): \emph{\bibinfo{title}{{When Only Topology Matters}}}.
\newblock {\sl \bibinfo{journal}{arXiv preprint arXiv:2102.03178}},
  \doi{10.48550/arXiv.2102.03178}.

\bibitemdeclare{article}{ChuJ2021superconductingand}
\bibitem{ChuJ2021superconductingand}
\bibinfo{author}{Ji~\surnamestart Chu\surnameend}, \bibinfo{author}{Xiaoyu
  \surnamestart He\surnameend}, \bibinfo{author}{Yuxuan \surnamestart
  Zhou\surnameend}, \bibinfo{author}{Jiahao \surnamestart Yuan\surnameend},
  \bibinfo{author}{Libo \surnamestart Zhang\surnameend}, \bibinfo{author}{Qihao
  \surnamestart Guo\surnameend}, \bibinfo{author}{Yongju \surnamestart
  Hai\surnameend}, \bibinfo{author}{Zhikun \surnamestart Han\surnameend},
  \bibinfo{author}{Chang-Kang \surnamestart Hu\surnameend},
  \bibinfo{author}{Wenhui \surnamestart Huang\surnameend}, \bibinfo{author}{Hao
  \surnamestart Jia\surnameend}, \bibinfo{author}{Dawei \surnamestart
  Jiao\surnameend}, \bibinfo{author}{Sai \surnamestart Li\surnameend},
  \bibinfo{author}{Yang \surnamestart Liu\surnameend},
  \bibinfo{author}{Zhongchu \surnamestart Ni\surnameend}, \bibinfo{author}{Lifu
  \surnamestart Nie\surnameend}, \bibinfo{author}{Xianchuang \surnamestart
  Pan\surnameend}, \bibinfo{author}{Jiawei \surnamestart Qiu\surnameend},
  \bibinfo{author}{Weiwei \surnamestart Wei\surnameend},
  \bibinfo{author}{Wuerkaixi \surnamestart Nuerbolati\surnameend},
  \bibinfo{author}{Zusheng \surnamestart Yang\surnameend},
  \bibinfo{author}{Jiajian \surnamestart Zhang\surnameend},
  \bibinfo{author}{Zhida \surnamestart Zhang\surnameend},
  \bibinfo{author}{Wanjing \surnamestart Zou\surnameend},
  \bibinfo{author}{Yuanzhen \surnamestart Chen\surnameend},
  \bibinfo{author}{Xiaowei \surnamestart Deng\surnameend},
  \bibinfo{author}{Xiuhao \surnamestart Deng\surnameend}, \bibinfo{author}{Ling
  \surnamestart Hu\surnameend}, \bibinfo{author}{Jian \surnamestart
  Li\surnameend}, \bibinfo{author}{Song \surnamestart Liu\surnameend},
  \bibinfo{author}{Yao \surnamestart Lu\surnameend}, \bibinfo{author}{Jingjing
  \surnamestart Niu\surnameend}, \bibinfo{author}{Dian \surnamestart
  Tan\surnameend}, \bibinfo{author}{Yuan \surnamestart Xu\surnameend},
  \bibinfo{author}{Tongxing \surnamestart Yan\surnameend},
  \bibinfo{author}{Youpeng \surnamestart Zhong\surnameend},
  \bibinfo{author}{Fei \surnamestart Yan\surnameend}, \bibinfo{author}{Xiaoming
  \surnamestart Sun\surnameend} \& \bibinfo{author}{Dapeng \surnamestart
  Yu\surnameend} (\bibinfo{year}{2023}): \emph{\bibinfo{title}{Scalable
  Algorithm Simplification Using Quantum {{AND}} Logic}}.
\newblock {\sl \bibinfo{journal}{Nature Physics}}
  \bibinfo{volume}{19}(\bibinfo{number}{1}), pp. \bibinfo{pages}{126--131},
  \doi{10.1038/s41567-022-01813-7}.

\bibitemdeclare{article}{CoeckeB2011interacting}
\bibitem{CoeckeB2011interacting}
\bibinfo{author}{Bob \surnamestart Coecke\surnameend} \& \bibinfo{author}{Ross
  \surnamestart Duncan\surnameend} (\bibinfo{year}{2011}):
  \emph{\bibinfo{title}{Interacting quantum observables: categorical algebra
  and diagrammatics}}.
\newblock {\sl \bibinfo{journal}{New Journal of Physics}}
  \bibinfo{volume}{13}(\bibinfo{number}{4}), p. \bibinfo{pages}{043016},
  \doi{10.1088/1367-2630/13/4/043016}.

\bibitemdeclare{inproceedings}{phaseGadgetSynth}
\bibitem{phaseGadgetSynth}
\bibinfo{author}{Alexander \surnamestart Cowtan\surnameend},
  \bibinfo{author}{Silas \surnamestart Dilkes\surnameend},
  \bibinfo{author}{Ross \surnamestart Duncan\surnameend}, \bibinfo{author}{Will
  \surnamestart Simmons\surnameend} \& \bibinfo{author}{Seyon \surnamestart
  Sivarajah\surnameend} (\bibinfo{year}{2020}): \emph{\bibinfo{title}{{Phase
  Gadget Synthesis for Shallow Circuits}}}.
\newblock In \bibinfo{editor}{Bob \surnamestart Coecke\surnameend} \&
  \bibinfo{editor}{Matthew \surnamestart Leifer\surnameend}, editors: {\sl
  \bibinfo{booktitle}{Proceedings 16th International Conference on Quantum
  Physics and Logic, Chapman University, Orange, CA, USA., 10-14 June 2019}},
  {\sl \bibinfo{series}{Electronic Proceedings in Theoretical Computer
  Science}} \bibinfo{volume}{318}, \bibinfo{publisher}{Open Publishing
  Association}, pp. \bibinfo{pages}{213--228}, \doi{10.4204/EPTCS.318.13}.

\bibitemdeclare{article}{cowtan2020generic}
\bibitem{cowtan2020generic}
\bibinfo{author}{Alexander \surnamestart Cowtan\surnameend},
  \bibinfo{author}{Will \surnamestart Simmons\surnameend} \&
  \bibinfo{author}{Ross \surnamestart Duncan\surnameend}
  (\bibinfo{year}{2020}): \emph{\bibinfo{title}{{A Generic Compilation Strategy
  for the Unitary Coupled Cluster Ansatz}}}.
\newblock {\sl \bibinfo{journal}{arXiv preprint arXiv:2007.10515}},
  \doi{10.48550/arXiv.2007.10515}.

\bibitemdeclare{article}{CozzolinoD2019quditcommunication}
\bibitem{CozzolinoD2019quditcommunication}
\bibinfo{author}{Daniele \surnamestart Cozzolino\surnameend},
  \bibinfo{author}{Beatrice \surnamestart Da~Lio\surnameend},
  \bibinfo{author}{Davide \surnamestart Bacco\surnameend} \&
  \bibinfo{author}{Leif~Katsuo \surnamestart Oxenløwe\surnameend}
  (\bibinfo{year}{2019}): \emph{\bibinfo{title}{{High-Dimensional Quantum
  Communication: Benefits, Progress, and Future Challenges}}}.
\newblock {\sl \bibinfo{journal}{Advanced Quantum Technologies}}
  \bibinfo{volume}{2}(\bibinfo{number}{12}), p. \bibinfo{pages}{1900038},
  \doi{10.1002/qute.201900038}.

\bibitemdeclare{article}{CuiS2017Diagonalhierarchy}
\bibitem{CuiS2017Diagonalhierarchy}
\bibinfo{author}{Shawn~X. \surnamestart Cui\surnameend},
  \bibinfo{author}{Daniel \surnamestart Gottesman\surnameend} \&
  \bibinfo{author}{Anirudh \surnamestart Krishna\surnameend}
  (\bibinfo{year}{2017}): \emph{\bibinfo{title}{{Diagonal gates in the Clifford
  hierarchy}}}.
\newblock {\sl \bibinfo{journal}{Phys. Rev. A}} \bibinfo{volume}{95}, p.
  \bibinfo{pages}{012329}, \doi{10.1103/PhysRevA.95.012329}.

\bibitemdeclare{article}{CuiS2015universalweakly}
\bibitem{CuiS2015universalweakly}
\bibinfo{author}{Shawn~X. \surnamestart Cui\surnameend},
  \bibinfo{author}{Seung-Moon \surnamestart Hong\surnameend} \&
  \bibinfo{author}{Zhenghan \surnamestart Wang\surnameend}
  (\bibinfo{year}{2015}): \emph{\bibinfo{title}{Universal quantum computation
  with weakly integral anyons}}.
\newblock {\sl \bibinfo{journal}{Quantum Information Processing}}
  \bibinfo{volume}{14}(\bibinfo{number}{8}), p. \bibinfo{pages}{2687–2727},
  \doi{10.1007/s11128-015-1016-y}.

\bibitemdeclare{article}{CuiS2015universalmetaplectic}
\bibitem{CuiS2015universalmetaplectic}
\bibinfo{author}{Shawn~X. \surnamestart Cui\surnameend} \&
  \bibinfo{author}{Zhenghan \surnamestart Wang\surnameend}
  (\bibinfo{year}{2015}): \emph{\bibinfo{title}{Universal quantum computation
  with metaplectic anyons}}.
\newblock {\sl \bibinfo{journal}{Journal of Mathematical Physics}}
  \bibinfo{volume}{56}(\bibinfo{number}{3}), p. \bibinfo{pages}{032202},
  \doi{10.1063/1.4914941}.

\bibitemdeclare{article}{DiY2013synthesis}
\bibitem{DiY2013synthesis}
\bibinfo{author}{Yao-Min \surnamestart Di\surnameend} \&
  \bibinfo{author}{Hai-Rui \surnamestart Wei\surnameend}
  (\bibinfo{year}{2013}): \emph{\bibinfo{title}{Synthesis of multivalued
  quantum logic circuits by elementary gates}}.
\newblock {\sl \bibinfo{journal}{Phys. Rev. A}} \bibinfo{volume}{87}, p.
  \bibinfo{pages}{012325}, \doi{10.1103/PhysRevA.87.012325}.

\bibitemdeclare{article}{GilesB2013multiqubitcliffordplustsynthesis}
\bibitem{GilesB2013multiqubitcliffordplustsynthesis}
\bibinfo{author}{Brett \surnamestart Giles\surnameend} \&
  \bibinfo{author}{Peter \surnamestart Selinger\surnameend}
  (\bibinfo{year}{2013}): \emph{\bibinfo{title}{Exact synthesis of multiqubit
  {Clifford+T} circuits}}.
\newblock {\sl \bibinfo{journal}{Physical Review A}}
  \bibinfo{volume}{87}(\bibinfo{number}{3}), \doi{10.1103/physreva.87.032332}.

\bibitemdeclare{inproceedings}{GlaudellA2022qutritmetaplecticsubset}
\bibitem{GlaudellA2022qutritmetaplecticsubset}
\bibinfo{author}{Andrew~N. \surnamestart Glaudell\surnameend},
  \bibinfo{author}{Neil~J. \surnamestart Ross\surnameend},
  \bibinfo{author}{John \surnamestart van~de Wetering\surnameend} \&
  \bibinfo{author}{Lia \surnamestart Yeh\surnameend} (\bibinfo{year}{2022}):
  \emph{\bibinfo{title}{{Qutrit Metaplectic Gates Are a Subset of
  Clifford+T}}}.
\newblock In \bibinfo{editor}{Fran\c{c}ois \surnamestart Le~Gall\surnameend} \&
  \bibinfo{editor}{Tomoyuki \surnamestart Morimae\surnameend}, editors: {\sl
  \bibinfo{booktitle}{17th Conference on the Theory of Quantum Computation,
  Communication and Cryptography (TQC 2022)}}, {\sl \bibinfo{series}{Leibniz
  International Proceedings in Informatics (LIPIcs)}} \bibinfo{volume}{232},
  \bibinfo{publisher}{Schloss Dagstuhl -- Leibniz-Zentrum f{\"u}r Informatik},
  \bibinfo{address}{Dagstuhl, Germany}, pp. \bibinfo{pages}{12:1--12:15},
  \doi{10.4230/LIPIcs.TQC.2022.12}.

\bibitemdeclare{article}{GokhaleP2019asymptotic}
\bibitem{GokhaleP2019asymptotic}
\bibinfo{author}{Pranav \surnamestart Gokhale\surnameend},
  \bibinfo{author}{Jonathan~M. \surnamestart Baker\surnameend},
  \bibinfo{author}{Casey \surnamestart Duckering\surnameend},
  \bibinfo{author}{Natalie~C. \surnamestart Brown\surnameend},
  \bibinfo{author}{Kenneth~R. \surnamestart Brown\surnameend} \&
  \bibinfo{author}{Frederic~T. \surnamestart Chong\surnameend}
  (\bibinfo{year}{2019}): \emph{\bibinfo{title}{Asymptotic improvements to
  quantum circuits via qutrits}}.
\newblock {\sl \bibinfo{journal}{Proceedings of the 46th International
  Symposium on Computer Architecture}}, \doi{10.1145/3307650.3322253}.

\bibitemdeclare{misc}{GongX2017equivalence}
\bibitem{GongX2017equivalence}
\bibinfo{author}{Xiaoyan \surnamestart Gong\surnameend} \&
  \bibinfo{author}{Quanlong \surnamestart Wang\surnameend}
  (\bibinfo{year}{2017}): \emph{\bibinfo{title}{{Equivalence of Local
  Complementation and Euler Decomposition in the Qutrit ZX-calculus}}},
  \doi{10.48550/arXiv.1704.05955}.

\bibitemdeclare{article}{GossetD2014opttoffolit}
\bibitem{GossetD2014opttoffolit}
\bibinfo{author}{David \surnamestart Gosset\surnameend}, \bibinfo{author}{Vadym
  \surnamestart Kliuchnikov\surnameend}, \bibinfo{author}{Michele \surnamestart
  Mosca\surnameend} \& \bibinfo{author}{Vincent \surnamestart Russo\surnameend}
  (\bibinfo{year}{2014}): \emph{\bibinfo{title}{{An Algorithm for the
  T-Count}}}.
\newblock {\sl \bibinfo{journal}{Quantum Info. Comput.}}
  \bibinfo{volume}{14}(\bibinfo{number}{15-16}), pp.
  \bibinfo{pages}{1261--1276}, \doi{10.5555/2685179.2685180}.

\bibitemdeclare{article}{GottesmanD1999ftqudit}
\bibitem{GottesmanD1999ftqudit}
\bibinfo{author}{Daniel \surnamestart Gottesman\surnameend}
  (\bibinfo{year}{1999}): \emph{\bibinfo{title}{{Fault-Tolerant Quantum
  Computation with Higher-Dimensional Systems}}}.
\newblock {\sl \bibinfo{journal}{Chaos, Solitons \& Fractals}}
  \bibinfo{volume}{10}(\bibinfo{number}{10}), p. \bibinfo{pages}{1749–1758},
  \doi{10.1016/s0960-0779(98)00218-5}.

\bibitemdeclare{article}{deGriendA2020architecturephasepoly}
\bibitem{deGriendA2020architecturephasepoly}
\bibinfo{author}{Arianne Meijer-van \surnamestart de~Griend\surnameend} \&
  \bibinfo{author}{Ross \surnamestart Duncan\surnameend}
  (\bibinfo{year}{2020}): \emph{\bibinfo{title}{Architecture-aware synthesis of
  phase polynomials for {NISQ} devices}}.
\newblock {\sl \bibinfo{journal}{ArXiv Preprint arxiv:2004.06052}},
  \doi{10.48550/ARXIV.2004.06052}.

\bibitemdeclare{inproceedings}{hadzihasanovic2015diagrammatic}
\bibitem{hadzihasanovic2015diagrammatic}
\bibinfo{author}{Amar \surnamestart Hadzihasanovic\surnameend}
  (\bibinfo{year}{2015}): \emph{\bibinfo{title}{A diagrammatic axiomatisation
  for qubit entanglement}}.
\newblock In: {\sl \bibinfo{booktitle}{2015 30th Annual ACM/IEEE Symposium on
  Logic in Computer Science}}, \bibinfo{organization}{IEEE}, pp.
  \bibinfo{pages}{573--584}, \doi{10.1109/LICS.2015.59}.

\bibitemdeclare{article}{HeyfronL2019quditcompiler}
\bibitem{HeyfronL2019quditcompiler}
\bibinfo{author}{Luke~E. \surnamestart Heyfron\surnameend} \&
  \bibinfo{author}{Earl \surnamestart Campbell\surnameend}
  (\bibinfo{year}{2019}): \emph{\bibinfo{title}{A quantum compiler for qudits
  of prime dimension greater than 3}}.
\newblock {\sl \bibinfo{journal}{ArXiv Preprint arxiv:1902.05634}},
  \doi{10.48550/arXiv.1902.05634}.

\bibitemdeclare{article}{HillA2021doublycontrolled}
\bibitem{HillA2021doublycontrolled}
\bibinfo{author}{Alexander~D. \surnamestart Hill\surnameend},
  \bibinfo{author}{Mark~J. \surnamestart Hodson\surnameend},
  \bibinfo{author}{Nicolas \surnamestart Didier\surnameend} \&
  \bibinfo{author}{Matthew~J. \surnamestart Reagor\surnameend}
  (\bibinfo{year}{2021}): \emph{\bibinfo{title}{Realization of arbitrary
  doubly-controlled quantum phase gates}}.
\newblock \doi{10.48550/arXiv.2108.01652}.

\bibitemdeclare{article}{HowardM2017resourcetheorymagic}
\bibitem{HowardM2017resourcetheorymagic}
\bibinfo{author}{Mark \surnamestart Howard\surnameend} \& \bibinfo{author}{Earl
  \surnamestart Campbell\surnameend} (\bibinfo{year}{2017}):
  \emph{\bibinfo{title}{{Application of a Resource Theory for Magic States to
  Fault-Tolerant Quantum Computing}}}.
\newblock {\sl \bibinfo{journal}{Phys. Rev. Lett.}} \bibinfo{volume}{118}, p.
  \bibinfo{pages}{090501}, \doi{10.1103/PhysRevLett.118.090501}.

\bibitemdeclare{article}{HowardM2012quditTgate}
\bibitem{HowardM2012quditTgate}
\bibinfo{author}{Mark \surnamestart Howard\surnameend} \& \bibinfo{author}{Jiri
  \surnamestart Vala\surnameend} (\bibinfo{year}{2012}):
  \emph{\bibinfo{title}{Qudit versions of the qubit $\ensuremath{\pi}/8$
  gate}}.
\newblock {\sl \bibinfo{journal}{Phys. Rev. A}} \bibinfo{volume}{86}, p.
  \bibinfo{pages}{022316}, \doi{10.1103/PhysRevA.86.022316}.

\bibitemdeclare{article}{KhanF2006synthesisqudit}
\bibitem{KhanF2006synthesisqudit}
\bibinfo{author}{Faisal~Shah \surnamestart Khan\surnameend} \&
  \bibinfo{author}{Marek \surnamestart Perkowski\surnameend}
  (\bibinfo{year}{2006}): \emph{\bibinfo{title}{Synthesis of multi-qudit hybrid
  and d-valued quantum logic circuits by decomposition}}.
\newblock {\sl \bibinfo{journal}{Theoretical Computer Science}}
  \bibinfo{volume}{367}(\bibinfo{number}{3}), pp. \bibinfo{pages}{336--346},
  \doi{10.1016/j.tcs.2006.09.006}.

\bibitemdeclare{article}{KissingerA2020reduc}
\bibitem{KissingerA2020reduc}
\bibinfo{author}{Aleks \surnamestart Kissinger\surnameend} \&
  \bibinfo{author}{John \surnamestart van~de Wetering\surnameend}
  (\bibinfo{year}{2020}): \emph{\bibinfo{title}{{Reducing the number of
  non-Clifford gates in quantum circuits}}}.
\newblock {\sl \bibinfo{journal}{Physical Review A}}
  \bibinfo{volume}{102}(\bibinfo{number}{2}),
  \doi{10.1103/physreva.102.022406}.

\bibitemdeclare{article}{GraphicalFourier2019}
\bibitem{GraphicalFourier2019}
\bibinfo{author}{Stach \surnamestart Kuijpers\surnameend},
  \bibinfo{author}{John \surnamestart van~de Wetering\surnameend} \&
  \bibinfo{author}{Aleks \surnamestart Kissinger\surnameend}
  (\bibinfo{year}{2019}): \emph{\bibinfo{title}{{Graphical Fourier Theory and
  the Cost of Quantum Addition}}}.
\newblock {\sl \bibinfo{journal}{arXiv preprint arXiv:1904.07551}},
  \doi{10.48550/arXiv.1904.07551}.

\bibitemdeclare{article}{MalletF2009qubitreadout2state}
\bibitem{MalletF2009qubitreadout2state}
\bibinfo{author}{François \surnamestart Mallet\surnameend},
  \bibinfo{author}{Florian~R. \surnamestart Ong\surnameend},
  \bibinfo{author}{Agustin \surnamestart Palacios-Laloy\surnameend},
  \bibinfo{author}{François \surnamestart Nguyen\surnameend},
  \bibinfo{author}{Patrice \surnamestart Bertet\surnameend},
  \bibinfo{author}{Denis \surnamestart Vion\surnameend} \&
  \bibinfo{author}{Daniel \surnamestart Esteve\surnameend}
  (\bibinfo{year}{2009}): \emph{\bibinfo{title}{Single-shot qubit readout in
  circuit quantum electrodynamics}}.
\newblock {\sl \bibinfo{journal}{Nature Physics}}
  \bibinfo{volume}{5}(\bibinfo{number}{11}), pp. \bibinfo{pages}{791--795},
  \doi{10.1038/nphys1400}.

\bibitemdeclare{article}{NikolaevaAS2022mctqutrit}
\bibitem{NikolaevaAS2022mctqutrit}
\bibinfo{author}{A.~S. \surnamestart Nikolaeva\surnameend},
  \bibinfo{author}{E.~O. \surnamestart Kiktenko\surnameend} \&
  \bibinfo{author}{A.~K. \surnamestart Fedorov\surnameend}
  (\bibinfo{year}{2022}): \emph{\bibinfo{title}{Decomposing the generalized
  {T}offoli gate with qutrits}}.
\newblock {\sl \bibinfo{journal}{Physical Review A}}
  \bibinfo{volume}{105}(\bibinfo{number}{3}),
  \doi{10.1103/physreva.105.032621}.

\bibitemdeclare{misc}{PetitL2020silicontwoqubitgates}
\bibitem{PetitL2020silicontwoqubitgates}
\bibinfo{author}{L.~\surnamestart Petit\surnameend},
  \bibinfo{author}{M.~\surnamestart Russ\surnameend}, \bibinfo{author}{H.~G.~J.
  \surnamestart Eenink\surnameend}, \bibinfo{author}{W.~I.~L. \surnamestart
  Lawrie\surnameend}, \bibinfo{author}{J.~S. \surnamestart Clarke\surnameend},
  \bibinfo{author}{L.~M.~K. \surnamestart Vandersypen\surnameend} \&
  \bibinfo{author}{M.~\surnamestart Veldhorst\surnameend}
  (\bibinfo{year}{2020}): \emph{\bibinfo{title}{{High-fidelity two-qubit gates
  in silicon above one Kelvin}}}, \doi{10.48550/ARXIV.2007.09034}.

\bibitemdeclare{article}{PinoJM2021honeywellarchitecture}
\bibitem{PinoJM2021honeywellarchitecture}
\bibinfo{author}{J.~M. \surnamestart Pino\surnameend}, \bibinfo{author}{J.~M.
  \surnamestart Dreiling\surnameend}, \bibinfo{author}{C.~\surnamestart
  Figgatt\surnameend}, \bibinfo{author}{J.~P. \surnamestart
  Gaebler\surnameend}, \bibinfo{author}{S.~A. \surnamestart Moses\surnameend},
  \bibinfo{author}{M.~S. \surnamestart Allman\surnameend},
  \bibinfo{author}{C.~H. \surnamestart Baldwin\surnameend},
  \bibinfo{author}{M.~\surnamestart Foss-Feig\surnameend},
  \bibinfo{author}{D.~\surnamestart Hayes\surnameend},
  \bibinfo{author}{K.~\surnamestart Mayer\surnameend},
  \bibinfo{author}{C.~\surnamestart Ryan-Anderson\surnameend} \&
  \bibinfo{author}{B.~\surnamestart Neyenhuis\surnameend}
  (\bibinfo{year}{2021}): \emph{\bibinfo{title}{Demonstration of the
  trapped-ion quantum {CCD} computer architecture}}.
\newblock {\sl \bibinfo{journal}{Nature}}
  \bibinfo{volume}{592}(\bibinfo{number}{7853}), pp. \bibinfo{pages}{209--213},
  \doi{10.1038/s41586-021-03318-4}.

\bibitemdeclare{article}{PrakashS2018normalform}
\bibitem{PrakashS2018normalform}
\bibinfo{author}{Shiroman \surnamestart Prakash\surnameend},
  \bibinfo{author}{Akalank \surnamestart Jain\surnameend},
  \bibinfo{author}{Bhakti \surnamestart Kapur\surnameend} \&
  \bibinfo{author}{Shubangi \surnamestart Seth\surnameend}
  (\bibinfo{year}{2018}): \emph{\bibinfo{title}{Normal form for single-qutrit
  {Clifford+T} operators and synthesis of single-qutrit gates}}.
\newblock {\sl \bibinfo{journal}{Physical Review A}}
  \bibinfo{volume}{98}(\bibinfo{number}{3}), \doi{10.1103/physreva.98.032304}.

\bibitemdeclare{inproceedings}{RanchinA2014quditzx}
\bibitem{RanchinA2014quditzx}
\bibinfo{author}{Andr{\'e} \surnamestart Ranchin\surnameend}
  (\bibinfo{year}{2014}): \emph{\bibinfo{title}{Depicting qudit quantum
  mechanics and mutually unbiased qudit theories}}.
\newblock In \bibinfo{editor}{Bob \surnamestart Coecke\surnameend},
  \bibinfo{editor}{Ichiro \surnamestart Hasuo\surnameend} \&
  \bibinfo{editor}{Prakash \surnamestart Panangaden\surnameend}, editors: {\sl
  \bibinfo{booktitle}{Proceedings of the 11th workshop on Quantum Physics and
  Logic, Kyoto, Japan, 4-6th June 2014}}, {\sl \bibinfo{series}{Electronic
  Proceedings in Theoretical Computer Science}} \bibinfo{volume}{172},
  \bibinfo{publisher}{Open Publishing Association}, pp.
  \bibinfo{pages}{68--91}, \doi{10.4204/EPTCS.172.6}.

\bibitemdeclare{article}{RingbauerM2021quditions}
\bibitem{RingbauerM2021quditions}
\bibinfo{author}{Martin \surnamestart Ringbauer\surnameend},
  \bibinfo{author}{Michael \surnamestart Meth\surnameend},
  \bibinfo{author}{Lukas \surnamestart Postler\surnameend},
  \bibinfo{author}{Roman \surnamestart Stricker\surnameend},
  \bibinfo{author}{Rainer \surnamestart Blatt\surnameend},
  \bibinfo{author}{Philipp \surnamestart Schindler\surnameend} \&
  \bibinfo{author}{Thomas \surnamestart Monz\surnameend}
  (\bibinfo{year}{2022}): \emph{\bibinfo{title}{A universal qudit quantum
  processor with trapped ions}}.
\newblock {\sl \bibinfo{journal}{Nature Physics}} \bibinfo{volume}{18}, pp.
  \bibinfo{pages}{1053--1057}, \doi{10.1038/s41567-022-01658-0}.

\bibitemdeclare{article}{SheldonS2016IBMCRgate}
\bibitem{SheldonS2016IBMCRgate}
\bibinfo{author}{Sarah \surnamestart Sheldon\surnameend},
  \bibinfo{author}{Easwar \surnamestart Magesan\surnameend},
  \bibinfo{author}{Jerry~M. \surnamestart Chow\surnameend} \&
  \bibinfo{author}{Jay~M. \surnamestart Gambetta\surnameend}
  (\bibinfo{year}{2016}): \emph{\bibinfo{title}{Procedure for systematically
  tuning up cross-talk in the cross-resonance gate}}.
\newblock {\sl \bibinfo{journal}{Phys. Rev. A}} \bibinfo{volume}{93}, p.
  \bibinfo{pages}{060302}, \doi{10.1103/PhysRevA.93.060302}.

\bibitemdeclare{article}{townsend-teague2021classifying}
\bibitem{townsend-teague2021classifying}
\bibinfo{author}{Alex \surnamestart Townsend-Teague\surnameend} \&
  \bibinfo{author}{Konstantinos \surnamestart Meichanetzidis\surnameend}
  (\bibinfo{year}{2021}): \emph{\bibinfo{title}{{Classifying Complexity with
  the ZX-Calculus: Jones Polynomials and Potts Partition Functions}}}.
\newblock {\sl \bibinfo{journal}{arXiv preprint arXiv:2103.06914}},
  \doi{10.48550/arXiv.2103.06914}.

\bibitemdeclare{inproceedings}{VilmartR2019nearminzx}
\bibitem{VilmartR2019nearminzx}
\bibinfo{author}{Renaud \surnamestart Vilmart\surnameend}
  (\bibinfo{year}{2019}): \emph{\bibinfo{title}{{A Near-Minimal Axiomatisation
  of ZX-Calculus for Pure Qubit Quantum Mechanics}}}.
\newblock In: {\sl \bibinfo{booktitle}{2019 34th Annual ACM/IEEE Symposium on
  Logic in Computer Science (LICS)}}, pp. \bibinfo{pages}{1--10},
  \doi{10.1109/LICS.2019.8785765}.

\bibitemdeclare{inproceedings}{WangQ2018qutrit}
\bibitem{WangQ2018qutrit}
\bibinfo{author}{Quanlong \surnamestart Wang\surnameend}
  (\bibinfo{year}{2018}): \emph{\bibinfo{title}{{Qutrit ZX-calculus is Complete
  for Stabilizer Quantum Mechanics}}}.
\newblock In \bibinfo{editor}{Bob \surnamestart Coecke\surnameend} \&
  \bibinfo{editor}{Aleks \surnamestart Kissinger\surnameend}, editors: {\sl
  \bibinfo{booktitle}{Proceedings 14th International Conference on Quantum
  Physics and Logic, Nijmegen, The Netherlands, 3-7 July 2017}}, {\sl
  \bibinfo{series}{Electronic Proceedings in Theoretical Computer Science}}
  \bibinfo{volume}{266}, \bibinfo{publisher}{Open Publishing Association}, pp.
  \bibinfo{pages}{58--70}, \doi{10.4204/EPTCS.266.3}.

\bibitemdeclare{inproceedings}{WangQ2014qutritcalculus}
\bibitem{WangQ2014qutritcalculus}
\bibinfo{author}{Quanlong \surnamestart Wang\surnameend} \&
  \bibinfo{author}{Xiaoning \surnamestart Bian\surnameend}
  (\bibinfo{year}{2014}): \emph{\bibinfo{title}{{Qutrit Dichromatic Calculus
  and Its Universality}}}.
\newblock In \bibinfo{editor}{Bob \surnamestart Coecke\surnameend},
  \bibinfo{editor}{Ichiro \surnamestart Hasuo\surnameend} \&
  \bibinfo{editor}{Prakash \surnamestart Panangaden\surnameend}, editors: {\sl
  \bibinfo{booktitle}{Proceedings of the 11th workshop on Quantum Physics and
  Logic, Kyoto, Japan, 4-6th June 2014}}, {\sl \bibinfo{series}{Electronic
  Proceedings in Theoretical Computer Science}} \bibinfo{volume}{172},
  \bibinfo{publisher}{Open Publishing Association}, pp.
  \bibinfo{pages}{92--101}, \doi{10.4204/EPTCS.172.7}.

\bibitemdeclare{article}{WangY2020quditsreview}
\bibitem{WangY2020quditsreview}
\bibinfo{author}{Yuchen \surnamestart Wang\surnameend}, \bibinfo{author}{Zixuan
  \surnamestart Hu\surnameend}, \bibinfo{author}{Barry~C. \surnamestart
  Sanders\surnameend} \& \bibinfo{author}{Sabre \surnamestart Kais\surnameend}
  (\bibinfo{year}{2020}): \emph{\bibinfo{title}{{Qudits and High-Dimensional
  Quantum Computing}}}.
\newblock {\sl \bibinfo{journal}{Frontiers in Physics}} \bibinfo{volume}{8}, p.
  \bibinfo{pages}{479}, \doi{10.3389/fphy.2020.589504}.

\bibitemdeclare{article}{vandewetering2020zxcalculus}
\bibitem{vandewetering2020zxcalculus}
\bibinfo{author}{John \surnamestart van~de Wetering\surnameend}
  (\bibinfo{year}{2020}): \emph{\bibinfo{title}{{ZX-calculus for the working
  quantum computer scientist}}}.
\newblock {\sl \bibinfo{journal}{arXiv preprint arXiv:2012.13966}},
  \doi{10.48550/arXiv.2012.13966}.

\bibitemdeclare{article}{vandeWeteringJ2021globalgates}
\bibitem{vandeWeteringJ2021globalgates}
\bibinfo{author}{John \surnamestart van~de Wetering\surnameend}
  (\bibinfo{year}{2021}): \emph{\bibinfo{title}{Constructing quantum circuits
  with global gates}}.
\newblock {\sl \bibinfo{journal}{New Journal of Physics}}
  \bibinfo{volume}{23}(\bibinfo{number}{4}), p. \bibinfo{pages}{043015},
  \doi{10.1088/1367-2630/abf1b3}.

\bibitemdeclare{article}{YeB2018cphasephoton}
\bibitem{YeB2018cphasephoton}
\bibinfo{author}{Biaoliang \surnamestart Ye\surnameend},
  \bibinfo{author}{Zhen-Fei \surnamestart Zheng\surnameend},
  \bibinfo{author}{Yu~\surnamestart Zhang\surnameend} \&
  \bibinfo{author}{Chui-Ping \surnamestart Yang\surnameend}
  (\bibinfo{year}{2018}): \emph{\bibinfo{title}{{Circuit QED: single-step
  realization of a multiqubit controlled phase gate with one microwave photonic
  qubit simultaneously controlling $n-1$ microwave photonic qubits}}}.
\newblock {\sl \bibinfo{journal}{Optics Express}}
  \bibinfo{volume}{26}(\bibinfo{number}{23}), p. \bibinfo{pages}{30689},
  \doi{10.1364/oe.26.030689}.

\bibitemdeclare{inproceedings}{YehL2022qutritcontrolledcliffordplust}
\bibitem{YehL2022qutritcontrolledcliffordplust}
\bibinfo{author}{Lia \surnamestart Yeh\surnameend} \& \bibinfo{author}{John
  \surnamestart van~de Wetering\surnameend} (\bibinfo{year}{2022}):
  \emph{\bibinfo{title}{{Constructing All Qutrit Controlled Clifford+T gates in
  Clifford+T}}}.
\newblock In \bibinfo{editor}{Claudio~Antares \surnamestart Mezzina\surnameend}
  \& \bibinfo{editor}{Krzysztof \surnamestart Podlaski\surnameend}, editors:
  {\sl \bibinfo{booktitle}{Reversible Computation}},
  \bibinfo{publisher}{Springer International Publishing},
  \bibinfo{address}{Cham}, pp. \bibinfo{pages}{28--50},
  \doi{10.1007/978-3-031-09005-9\_3}.

\bibitemdeclare{article}{YurtalanM2020Walsh-Hadamard}
\bibitem{YurtalanM2020Walsh-Hadamard}
\bibinfo{author}{M.~A. \surnamestart Yurtalan\surnameend},
  \bibinfo{author}{J.~\surnamestart Shi\surnameend},
  \bibinfo{author}{M.~\surnamestart Kononenko\surnameend},
  \bibinfo{author}{A.~\surnamestart Lupascu\surnameend} \&
  \bibinfo{author}{S.~\surnamestart Ashhab\surnameend} (\bibinfo{year}{2020}):
  \emph{\bibinfo{title}{{Implementation of a Walsh-Hadamard Gate in a
  Superconducting Qutrit}}}.
\newblock {\sl \bibinfo{journal}{Phys. Rev. Lett.}} \bibinfo{volume}{125}, p.
  \bibinfo{pages}{180504}, \doi{10.1103/PhysRevLett.125.180504}.

\end{thebibliography}

\appendix

\section{Qutrit ZX-calculus}
\subsection{Necessity of rules}\label{app:necessity-of-rules}
We can show that most of the rules in Figure~\ref{fig:qutritphaseexactflexsymmetricrules} are \emph{necessary}, meaning that they cannot be derived from the other rules. We do this by adapting the reasoning of Ref.~\cite{VilmartR2019nearminzx}.
Namely, the following rules are definitely necessary:
\begin{itemize}
    \item $(SZ)$: this is the only rule which can decompose a generator with four or more legs into generators with fewer legs.
    \item $(P)$: this is the only rule which resolves diagrams containing generators to the identity.
    \item $(B1)$: this is the only rule that can transform a connected diagram into a disconnected one.
    \item $(EU)$: this is necessary per the argument of Ref.~\cite[Proposition 3.2]{WangQ2014qutritcalculus}.
    \item At least one of $(H)$ and $(H')$ is necessary as these are the only ones that can convert a diagram containing a $X$ generator with a non-integer phase into one containing a $Z$ generator with a non-integer phase.
\end{itemize}

We do not know whether the other rules are necessary, although we do suspect this is the case.

\subsection{Proofs of the derived rules}\label{appendix:derivedrules}

\begin{lemma}
    The $(ID)$ rule can be derived from the $(SZ)$ and $(SP)$ rules.
\end{lemma}
\begin{proof}~
    \[\input{XD_rules/derived_rules/id.tikz} \qedhere\]
\end{proof}

\begin{lemma}
    The $(H2)$ rule can be derived from the $(H')$ and $(ID)$ rules.
\end{lemma}
\begin{proof}~
    \[\input{XD_rules/derived_rules/h2.tikz} \qedhere\]
\end{proof}

\begin{lemma}
    The $(H4)$ rule can be derived from the $(H)$, $(ID)$, and $(H2)$ rules.
\end{lemma}
\begin{proof}~
    \[\input{XD_rules/derived_rules/h4.tikz} \qedhere\]
\end{proof}

\begin{lemma}
    The $(SX)$ rule can be derived from the $(SZ)$, $(H')$, $(H2)$, and $(H4)$ rules.
\end{lemma}
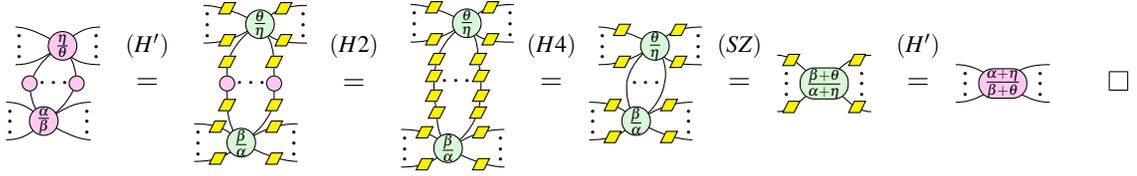
\begin{proof}
    \[\input{XD_rules/derived_rules/sx.tikz} \qedhere\]
\end{proof}

\begin{lemma}
    The $(P1')$ rules can be derived from the $(P1)$, $(SX)$, $(H)$, $(H')$, $(H2)$, and $(H4)$ rules.
\end{lemma}
\begin{proof}
    Let's first derive the rule for $x=0$:
    \begin{equation}\label{eq:k1_0}
        \input{XD_rules/derived_rules/k1_0.tikz}
    \end{equation}

    From that, we can derive the below rule:
    \begin{equation}    
        \input{XD_rules/derived_rules/k1_2.tikz}
    \end{equation}

    The two above rules, along with the $(P1)$ rule, are captured by the following rule where $x \in \{0,1,2\}$:
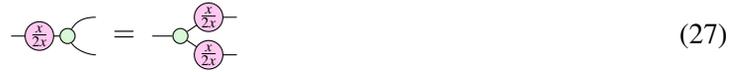
    \begin{equation}\label{eq:k1}
        \input{XD_rules/derived_rules/k1.tikz}
    \end{equation}

    We now colour-change the above rule to finish deriving all the $(P1')$ rules:
    \begin{equation}    
        \input{XD_rules/derived_rules/k1_p.tikz}
    \end{equation}
\end{proof}

\begin{lemma}
    The $(P2')$ rules can be derived from the $(P2)$, $(H)$, $(H')$, $(SZ)$, $(SX)$, $(H2)$, and $(H4)$ rules.
\end{lemma}
\begin{proof}
    We prove them one by one:
    \begin{equation}
        \input{XD_rules/derived_rules/k2_2.tikz}
    \end{equation}

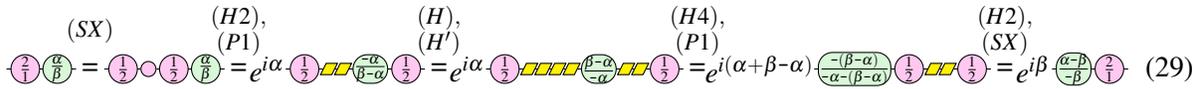
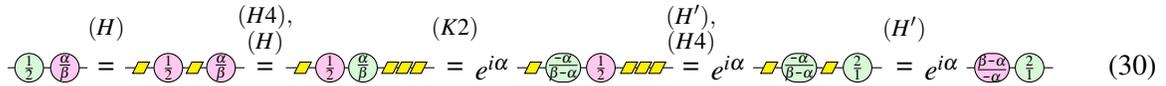
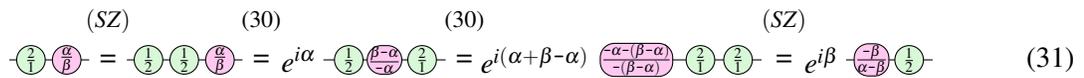
    \begin{equation}\label{eq:k2_p1}
        \input{XD_rules/derived_rules/k2_p1.tikz}
    \end{equation}

    \begin{equation}
        \input{XD_rules/derived_rules/k2_p2.tikz}
    \end{equation}

\end{proof}

\begin{lemma}
    The $(EU')$ rules can be derived from the $(EU)$, $(H)$, $(H')$, $(SZ)$, $(SX)$, $(H2)$, and $(H4)$ rules.
\end{lemma}
\begin{proof}
    We show the first equation directly:
    \begin{equation}\label{eq:eu_p}
        \input{XD_rules/derived_rules/eu_p.tikz}
    \end{equation}
    
    For the second one we first note that:
    \begin{equation}\label{eq:eu_p1}
        \input{XD_rules/derived_rules/eu_p1.tikz}
    \end{equation}
    
    Then:
    \begin{equation}\label{eq:eu1_p1}
        \input{XD_rules/derived_rules/eu1_p1.tikz}
    \end{equation}

    And finally, we find the different decomposition of $H$:
    \begin{equation}\label{eq:eu1}
        \input{XD_rules/derived_rules/eu1.tikz}
    \end{equation}
\end{proof}

\section{Constructing general phase multipliers}\label{app:phase-multipliers}

When we have two variables we use the formula
\begin{equation}\label{app:trit-plus-rel}
    x\cdot y~\text{mod } 3 = (x^2~\text{mod } 3) + (y^2~\text{mod } 3) - ((x+y)^2~\text{mod } 3)
\end{equation}
to construct the two-qutrit phase multiplier. We generalised this to three variables in the following way:
\begin{equation}\label{app:three-variables}
    (x\cdot y)\cdot z\  =\  x^2\cdot z + y^2\cdot z - (x+y)^2\cdot z
    \ =\  x^2 + y^2 + z^2 - (x^2 + z)^2 - (y^2 + z)^2 - (x+y)^2 + ((x+y)^2+z)^2.
\end{equation}
To see how we go to $4$ variables and beyond, we start with the expression $(x\cdot y\cdot z)\cdot w$ and decompose $x\cdot y\cdot z$ with the above formula resulting in terms $t_1^2,\ldots, t_n^2$. Each of these terms is a square, because that is the case for all the terms in Eq.~\eqref{app:trit-plus-rel}. Since we are working with qutrits we have $(t_j^2)^2 = t_j^2$. 
The terms in our formula are now of the form $t_j^2\cdot w$. We apply Eq.~\eqref{app:trit-plus-rel} to each of these. This gives us terms $t_j^4$, $w^2$ and $(t_j^2+w)^2$. The first of these is just $t_j^2$, and by induction we already know how to construct the appropriate phase term on the circuit for this term. The second of these is $w^2$, and hence corresponds to a simple phase gate. Note that this is the same for each $t_j^2\cdot w$ we are decomposing. Furthermore, the plus signs and minus signs on the terms are such that they almost all cancel, and we will have one copy of $w^2$. 

The only `interesting' new term we then get is hence $(t_j^2+w)^2$. For instance, in Eq.~\eqref{app:three-variables} the terms of this form are $(x^2+z)^2$, $(y^2+z)^2$ and $((x+y)^2 +z)^2$. The corresponding phase terms are constructed by using the gadget of Eq.~\eqref{eq:square-control} to store $t_j^2+w$ ``on the wire'' and then applying a $Z(\alpha,\alpha)$ phase gate.

Hence, if we go from 3 to 4 variables we get each of the original terms $t_j^2$, plus a $w^2$ term and a $(t_j^2 + w)^2$ term for each $j$. This straightforwardly generalises to $n$ variables, and it is then easy to check that we will have $2^n-1$ terms.

We can build a circuit for the $n>2$ qutrit phase multiplier by first building the circuit for $n-1$ qutrits, and then inserting the gadget of Eq.~\eqref{eq:square-phase} after every application of a $Z(\alpha,\alpha)$ phase with as the target the $n$th qutrit. The reason this works is because the $n$-qutrit phase multiplier still contains every term of the $n-1$ qutrit multiplier, but now also needs to combine those terms with the $n$th variable. In the four variable case, we would first store on a wire the value of the term $t_j$ we need, and then apply a $Z(\alpha,\alpha)$ gate in order to get the phase $e^{i\alpha t_j^2}$. Then we would apply Eq.~\eqref{eq:square-phase} on the qutrit of $w$ in order to get the phase $e^{i\alpha (t_j^2+w)^2}$. This construction involves temporarily storing $t_j^2+w$ on the wire of $w$, so we can use this term if we want to construct the five-qutrit phase multiplier as well.

We then see that the cost of the $n$-qutrit phase multiplier in terms of (non-Clifford) gates is the cost of the $n-1$ qutrit phase multiplier plus the cost of $2^{n-1}-1$ applications of the Eq.~\eqref{eq:square-phase} gadget. In particular, each phase term requires precisely one of either a $Z(\alpha,\alpha)$ or a $Z(-\alpha,-\alpha)$ gate, so that we need $2^n-1$ of them. The circuit of Eq.~\eqref{eq:square-phase} requires 6 $T$ gates to construct, and hence the $T$-count of the $n$-qutrit phase multiplier is $6 (2^{n-1}-1) = 3\cdot 2^n - 6$ (for $n>2$).

\end{document}